\documentclass[12pt]{amsart}
\usepackage{graphicx,amscd, color,amsmath,amsfonts,amssymb,geometry,amssymb}
\usepackage[initials]{amsrefs}

\newtheorem{theorem}{Theorem}[section]
\newtheorem{proposition}[theorem]{Proposition}
\newtheorem{corollary}[theorem]{Corollary}
\newtheorem{example}[theorem]{Example}
\newtheorem{examples}[theorem]{Examples}
\newtheorem{remark}[theorem]{Remark}

\newtheorem{lemma}[theorem]{Lemma}

\newtheorem{definition}[theorem]{Definition}
\numberwithin{equation}{section}

\begin{document}
\title{ Completely reducible maps in Quantum Information Theory}

\hyphenation{pro-per-ties}

\author[Cariello ]{D. Cariello}
%\thanks{PACS numbers: 03.65.Ud, 03.67.Mn}

\address{Faculdade de Matem\'atica, \newline\indent Universidade Federal de Uberl\^{a}ndia, \newline\indent 38.400-902 Ð Uberl\^{a}ndia, Brazil.}\email{dcariello@famat.ufu.br}

\address{Departamento de An\'{a}lisis Matem\'{a}tico,\newline\indent Facultad de Ciencias Matem\'{a}ticas, \newline\indent Plaza de Ciencias 3, \newline\indent Universidad Complutense de Madrid,\newline\indent Madrid, 28040, Spain.}
\email{dcariell@ucm.es}



\maketitle

\begin{abstract} In order to compute the Schmidt decomposition of  $A\in M_k\otimes M_m$, we must consider an associated self-adjoint map.
Here, we show that if  $A$ is positive under partial transposition (PPT) or symmetric with positive coefficients (SPC) or invariant under realignment then its associated self-adjoint map is completely reducible. We give applications of this fact in Quantum Information Theory. We recover some theorems recently proved for PPT and SPC matrices and we prove these theorems for matrices invariant under realignment using theorems of Perron-Frobenius theory. We also provide a new proof of the fact that if $\mathbb{C}^{k}$ contains $k$ mutually unbiased bases then $\mathbb{C}^{k}$ contains $k+1$. We search  for other types of matrices  that could have the same property.

We consider a group of linear transformations acting on $M_k\otimes M_k$, which contains the partial transpositions and the realignment map. For each element of this group, we consider the set of matrices in $M_k\otimes M_k\simeq M_{k^2}$ that are positive and remain positive, or invariant, under the action of this element. Within this family of sets, we have the set of PPT matrices, the set of  SPC matrices and the set of  matrices invariant under realignment. We show that these three sets are the only sets of this family such that the associated self-adjoint map of each matrix is completely reducible. We also show that every matrix invariant under realignment is PPT in $M_2\otimes M_2$ and we present a counterexample in $M_k\otimes M_k$, $k\geq 3$.

\end{abstract}

\section*{Introduction}
Let $M_k$ denote the set of complex matrices of order $k$. Let us denote by $VM_kW$ the set $\{VXW,\ X\in M_k\}$, where $V,W\in M_k$ are orthogonal projections. If $V=W$ then $VM_kV$ is a hereditary finite dimensional $C^*$-algebra (see \cite{evans}). Let $P_k$ denote the set of positive semidefinite Hermitian matrices in $M_k$. A linear transformation $L:VM_kV\rightarrow WM_mW$ is said to be a positive map if $L(P_k\cap VM_kV)\subset P_m\cap WM_mW$.  A non-null positive map $L:VM_kV\rightarrow VM_kV$ is called irreducible, if for every $V'M_kV'\subset VM_kV$ such that $L(V'M_kV')\subset V'M_kV'$, we have $V'=V$ or $V'=0$.

Two well known theorems of Perron-Frobenius Theory are theorems 2.3 and 2.5 in \cite{evans}: If $L:VM_kV\rightarrow VM_kV$ is a positive map then exists $\gamma\in P_k\cap VM_kV$ such that $L(\gamma)=\lambda\gamma$, where $\lambda$ is the spectral radius of $L$. Moreover, if $L$ is irreducible then this eigenvalue has multiplicity $1$.
    
Let us say that $L:VM_kV\rightarrow VM_kV$  is completely reducible, if there are orthogonal hereditary subalgebras of $VM_kV$, invariant under $L$, such that $L$ restricted to each of these subalgebras is irreducible and $L$ restricted to the orthogonal complement of their direct sum is null (see definition \ref{definitioncompletelyreducible}).
The main theorems of this paper are related to the concept of completely reducible map and shall be obtained using the aforementioned  theorems of Perron-Frobenius theory.  This concept is related to the concept of completely reducible matrix (see \cite{schaefer}).  In order to describe these main theorems, let us identify  the tensor product space $M_{k}\otimes M_{m}$ with $M_{km}$, via Kronecker product.

%A symmetric matrix $A\in M_k$ with positive entries is called completely reducible, if  $PAP^{-1}$ is a block matrix such that the off-diagonal blocks  are null and the diagonal blocks are irreducible, for some permutation matrix $P$ (see chapter 1 in \cite{schaefer}).  

Suppose $A=\sum_{i=1}^nA_i\otimes B_i\in M_k\otimes M_m\simeq M_{km}$ is a positive semidefinite Hermitian matrix. Consider the maps $F_A: M_m\rightarrow M_k$, $F_A(X)=\sum_{i=1}^n tr(B_iX)A_i$ and $G_A: M_k\rightarrow M_m$, $G_A(X)=\sum_{i=1}^n tr(A_iX)B_i$. These are positive maps and adjoints with respect to the trace inner product. Thus, $F_A\circ G_A: M_k\rightarrow M_k$ is a self adjoint positive map.
Here, in section 4, we show that if $A$ is positive under partial transposition (PPT) or symmetric with positive coefficients (SPC) or invariant under realignment then $F_A\circ G_A: M_k\rightarrow M_k$ is completely reducible.
 There are applications of this fact in Quantum Information theory, as we describe bellow. 

First, if $F_A\circ G_A: M_k\rightarrow M_k$ is completely reducible then $A$ is a sum of weakly irreducible matrices (see definition \ref{definitionweaklyirreducible}) with support on orthogonal local Hilbert spaces, therefore $A$ is separable (see definition \ref{definitionseparability}) if and only if each weakly irreducible summand is separable (see corollary \ref{corollarysplitdecomposition}). This theorem was proved in \cite{cariello} for PPT and SPC matrices, thus we extend this theorem for matrices invariant under realignment. Notice that a necessary condition for a matrix to be separable is to be PPT, thus the author of \cite{cariello} reduced the separability problem to the weakly irreducible case. In \cite{cariello}, we can also find a description of weakly irreducible PPT or SPC matrix. Here, we obtained the same description for matrices invariant under realignment (proposition \ref{propositionweaklyirred}).

%Actually, using the description of weakly irreducible matrices invariant under realigment, we can go further. We obtain a necessary and sufficient condition for a positive semidefinite Hermitian matrix $A\in M_k\otimes M_m\simeq M_{km}$ to be weakly irreducible (see theorem ???????). 

Second, if $F_A\circ G_A: M_k\rightarrow M_k$ is completely reducible with eigenvalues $1$ or $0$ then $A$ is separable in a very strong sense (proposition \ref{propositioneigenvalues1}). 
Using this theorem for a matrix $A$ invariant under realignment (see proposition \ref{mainapplication}), we obtain a new proof of the following theorem proved in \cite{weiner}:
If $\mathbb{C}^k$ contains $k$ mutually unbiased bases then exists another orthonormal basis which is mutually unbiased with these $k$ bases $($theorem \ref{thelastbasis}$)$.

In Quantum Information Theory, the concept of mutually unbiased bases (definition \ref{defMUBS}) has been shown to be useful. It has applications in state determination, quantum state tomography, cryptography $($See \cite{ivan},\cite{wootters}, \cite{wootters2}, \cite{cerf}$)$.
It is known that $k+1$ is an upper bound for the number of mutually unbiased bases in $\mathbb{C}^k$ and the existence of this number of bases is an open problem, when $k$ is not a power of prime. When $k$ is a power of prime, some constructive methods were used to obtain these $k+1$ bases $($See \cite{ivan},\cite{wootters}, \cite{band}$)$.

The realigment map $($definition \ref{defS}$)$ is important in Quantum Information Theory by its use in the realigment criterion $($\cite{chen}, \cite{rudolph}$)$. 
This new proof of the existence of the last mutually unbiased basis  shows a connection between two current topics in Quantum Information theory, the realigment map and mutually unbiased bases.

We shall search for other types of $A$ such that $F_A\circ G_A: M_k\rightarrow M_k$ could be completely reducible using the following idea.
Let $S_4$ be the group of permutations of $\{1,2,3,4\}$ and let us use the cycle notation.
Let $\sigma\in S_4$ and let us define the linear map $L_{\sigma}: M_k\otimes M_k\rightarrow M_k\otimes M_k$, $L_{\sigma}(a_1a_2^t\otimes a_3a_4^t)=a_{\sigma(1)}a_{\sigma(2)}^t\otimes a_{\sigma(3)}a_{\sigma(4)}^t,$ for every $a_1,a_2,a_3,a_4\in\mathbb{C}^k$.  Notice that $\{L_{\sigma},\ \sigma\in S_4\}$ is a group acting on $M_k\otimes M_k$. 

Many well known maps in quantum information theory have the type $L_{\sigma}$. For example:
\begin{itemize}
\item[a)] $L_{(34)}$ is the partial transposition, $L_{(34)}(C\otimes D)=C\otimes D^t$,
\item[b)] $L_{(12)(34)}(A)=A^t$ is the transposition,
\item[c)]  $L_{(24)}(A)=AT$, where $T$ is the flip operator (see definition \ref{definition1}),
\item[d)]  $L_{(13)(24)}(C\otimes D)=D\otimes C$, 
\item[e)] $L_{(23)}$ is the realignment map,
\item[f)] $L_{(243)}$ is the partial transposition composed with the realignment map.
\end{itemize}

This group is also known in quantum information theory for providing criterions to detect entanglement. For each element of this group exists a corresponding criterion analogous to the realignment criterion $($see \cite{horodecki1},\cite{horodecki2}$)$.

Notice that, by definition \ref{defPPT}, PPT matrices are positive semidefinite Hermitian matrices that remain positive under partial transposition $($the map $L_{(34)})$. Here, we show that SPC matrices  are positive semidefinite Hermitian matrices that remain positive under partial transposition composed with the realigment map $($the map $L_{(243)})$ $($lemma \ref{SPCequiv}$)$.
Let us define the  sets $$P_{\sigma}=\{A\in M_k\otimes M_k, A\ \text{and}\ L_{\sigma}(A)\ \text{are positive semidefinite Hermitian matrices}\}$$
and we shall consider the following problem: 
\begin{flushleft}
 \textbf{Problem 1:} Which of the sets $P_{\sigma}$ contains only matrices $A$ such that $F_A\circ G_A:M_k\rightarrow M_k$ is completely reducible?
\end{flushleft}

The results obtained in \cite{cariello} are related to our first problem. Indeed, they play a very special role in the solution. In section 6, we show that if all matrices $A$ of $P_{\sigma}$ are such that $F_A\circ G_A:M_k\rightarrow M_k$ is completely reducible then $P_{\sigma}=\{$PPT matrices$\}$ or $P_{\sigma}=\{$SPC matrices$\}$.

In our opinion, there is a lack of symmetry in this solution, because PPT matrices are matrices that remain positive under partial transposition and SPC matrices are matrices that remain positive under  partial transposition composed with the realigment map, but matrices that remain positive under the action of the realigment map may not have this property. For example: $uu^t$ is positive semidefinite $($definition \ref{definition1}$)$, $S(uu^t)=Id\otimes Id$ $($definition \ref{defS}$)$ and $F_{uu^t}\circ G_{uu^t}: M_k\rightarrow M_k$ is not completely reducible $($lemma \ref{notdecomposable}$)$. 
This lack of symmetry leads to our second problem. Define 
 $$I_{\sigma}=\{A\in M_k\otimes M_k, A\ \text{is positive semidefinite and}\ A=L_{\sigma}(A) \}$$
and let us consider the following problem:

\begin{flushleft}
 \textbf{Problem 2:} Which of the sets $I_{\sigma}$ contains only matrices $A$ such that $F_A\circ G_A:M_k\rightarrow M_k$ is completely reducible?
\end{flushleft}

In section 6, we provide the answer to the second problem. It is interesting to notice that the answer depends on the dimension $k$. 
For $k\geq3$, the only sets $I_{\sigma}$ in the solution of problem $2$ satisfy $I_{\sigma}\subset\{\text{PPT matrices}\}$ or $I_{\sigma}\subset\{\text{SPC matrices}\}$ or $I_{\sigma}\subset\{\text{Matrices Invariant under Realignment}\}$ $($theorem \ref{solutionproblem_2_k>2}$)$.
For $k=2$, the only sets $I_{\sigma}$ in the solution of problem $2$ satisfy $I_{\sigma}\subset\{\text{PPT matrices}\}$ $($theorem \ref{solutionproblem_2_k=2}$)$. In order to prove this last result, we show that every matrix invariant under realignment is PPT in $M_2\otimes M_2$ $($lemma \ref{invariantPPT}$)$ and we present a counterexample in $M_k\otimes M_k$, for $k\geq 3$ $($example \ref{counterexample}$)$. 

Some of these $I_{\sigma}$ were considered previously in papers related to the separability problem in Quantum Information Theory. For example,  the authors of  \cite{guhne} and  \cite{kraus} considered matrices invariant under multiplication by the flip operator and matrices invariant under partial transposition, respectively. 
 
This paper is arranged as follows. In section 1, we describe the definitions and the preliminary results that shall be used. In section 2, we describe some results concerning  completely reducible maps. These results are  consequences of  theorems of Perron-Frobenius Theory. In section 3, we assume $F_A\circ G_A:M_k\rightarrow M_k$ is completely reducible and we give two applications in Quantum Information Theory. We also provide an equivalent way to prove that $F_A\circ G_A:M_k\rightarrow M_k$ is completely reducible (lemma \ref{lemmaequivalence}). This equivalent way shall be used in section 4, in order to prove that  $F_A\circ G_A:M_k\rightarrow M_k$ is  completely reducible, if $A\in M_k\otimes M_m$ is PPT or SPC or invariant under realignment. We also provide examples of positive semidefinite Hermitian matrices $A\in M_k\otimes M_m$ such that  $F_A\circ G_A:M_k\rightarrow M_k$ is not completely reducible. In section 5, we show that if $\mathbb{C}^k$ contains $k$ mutually unbiased bases then $\mathbb{C}^k$ contains $k+1$ and we show that this last basis is unique up to multiplication by complex vectors of norm 1.  In section 6, we search for other types of  $A\in M_k\otimes M_k$ such that $F_A\circ G_A:M_k\rightarrow M_k$ could be completely reducible.

\section{Preliminary Results and Definitions}

Let $M_k$ denote the set of complex matrices of order $k$ and $\mathbb{C}^k$ be the set of column vectors with $k$ complex entries. We shall identify the tensor product space $\mathbb{C}^k\otimes\mathbb{C}^m$ with $\mathbb{C}^{km}$ and the tensor product space $M_{k}\otimes M_{m}$ with $M_{km}$, via Kronecker product (i.e., if $A=(a_{ij})\in M_k$ and $B\in M_m$ then $A\otimes B=(a_{ij}B)\in M_{km}$ and if $v=(v_i)\in \mathbb{C}^k, w\in \mathbb{C}^m$ then $v\otimes w=(v_iw)\in\mathbb{C}^{km}$).

The identification of the tensor product space $\mathbb{C}^k\otimes\mathbb{C}^m$ with $\mathbb{C}^{km}$ and the tensor product space $M_{k}\otimes M_{m}$ with $M_{km}$, via Kronecker product, allow us to write $(v\otimes w)(r\otimes s)^t= vr^t\otimes ws^t$, where $v\otimes w\in \mathbb{C}^k\otimes\mathbb{C}^m$ is a column, $(v\otimes w)^t$ its transpose and $v,r\in\mathbb{C}^k$ and $w,s\in\mathbb{C}^m$.  Therefore if $x,y\in\mathbb{C}^k\otimes\mathbb{C}^m\simeq\mathbb{C}^{km}$ we have $xy^t\in M_k\otimes M_m\simeq M_{km}$.
Here, $tr(A)$ denotes the trace of a matrix $A$, $\overline{A}$ stands for the matrix whose entries are $\overline{a_{ij}}$, where $\overline{a_{ij}}$ is the complex conjugate of the entry $a_{ij}$ of $A$ and $A^t$ stands for the transpose of $A$. We shall consider the usual inner product in $M_{k}$, $\langle A,B\rangle=tr(AB^*)$, and the usual inner product in $\mathbb{C}^k$, $\langle x,y\rangle=x^t\overline{y}$. If $A=\sum_{i=1}^nA_i\otimes B_i$, we shall denote by $A^{t_2}$ the matrix $\sum_{i=1}^nA_i\otimes B_i^t$, which is called the partial transposition of $A$. The image (or the range) of the matrix $A\in M_k\otimes M_m\simeq M_{km}$  in $\mathbb{C}^k\otimes\mathbb{C}^m\simeq \mathbb{C}^{km}$ shall be denoted by $\Im(A)$.

\begin{definition}\label{Hermitiandecomposition} A decomposition of a matrix $A\in M_{k}\otimes M_{m}\simeq M_{km}$, $\sum_{i=1}^n A_i\otimes B_i$, is a Hermitian decomposition if $A_i\in M_k$ and $B_i\in M_m$ are Hermitian matrices for every $i$.
\end{definition}

\begin{definition}\label{HermitianSchmidtdecomposition} A decomposition of a matrix A $\in M_{k}\otimes M_{m}$,
$\sum_{i=1}^n \lambda_i \gamma_i\otimes \delta_i$,
is a Schmidt decomposition if $\{\gamma_i|1\leq i\leq n\}\subset M_k$,\ $\{\delta_i|1\leq i\leq n\}\subset M_m$ are orthonormal sets with respect to the trace inner product,  $\lambda_i\in\mathbb{R}$ and $\lambda_i>0$.
Also, if $\gamma_i$ and $\delta_i$ are Hermitian matrices for every $i$, then $\sum_{i=1}^n \lambda_i \gamma_i\otimes \delta_i$ is a Hermitian Schmidt decomposition of A. 
\end{definition}

\begin{definition}\label{defPPT}$($\textbf{PPT matrices}$)$ Let $A=\sum_{i=1}^nA_i\otimes B_i\in M_k\otimes M_m\simeq M_{km}$ be a positive semidefinite Hermitian matrix. We say that $A$ is positive under partial transposition or simply  PPT, if $A^{t_2}=Id\otimes(\cdot)^t(A)=\sum_{i=1}^nA_i\otimes B_i^t$ is positive semidefinite.
\end{definition}

\begin{definition}\label{defSPC}$($\textbf{SPC matrices}$)$ Let $A\in M_k\otimes M_k\simeq M_{k^2}$ be a positive semidefinite Hermitian matrix. We say that $A$ is symmetric with positive coefficients or simply SPC, if $A$ has the following symmetric Hermitian Schmidt decomposition with positive coefficients: $\sum_{i=1}^n\lambda_i\gamma_i\otimes\gamma_i$, with $\lambda_i>0$, for every $i$.
\end{definition}

\begin{definition}\label{definitionseparability}\textbf{$($Separable Matrices$)$} Let $A\in M_k\otimes M_m$.
We say that $A$ is separable if $A$ can be approximated in norm by matrices of the following type: $\sum_{i=1}^n C_i\otimes D_i$ such that  $C_i\in M_k$ and $D_i\in M_m$ are positive semidefinite Hermitian matrices for every $i$.
\end{definition}

\begin{definition}\label{definition1}  Let  $\{e_1,\ldots,e_k\}$ be the canonical basis of $\mathbb{C}^k$. 
\begin{enumerate}
 
\item 
Let $T=\sum_{i,j=1}^ke_ie_j^t\otimes e_je_i^t\in M_k\otimes M_k\simeq M_{k^2}$. This matrix satisfies $Ta\otimes b=b\otimes a$, $(a\otimes b)^tT=(b\otimes a)^t$,  for every $a,b\in\mathbb{C}^k$, where $a\otimes b$ is a column vector in $\mathbb{C}^{k^2}$ and  $(b\otimes a)^t$ is its transpose. This matrix is usually called the flip operator $($see \cite{guhne}$)$. 
\item  Let $u=\sum_{i=1}^ke_i\otimes e_i\in \mathbb{C}^k\otimes\mathbb{C}^k$. 
\item Let $F:M_k\rightarrow \mathbb{C}^k\otimes\mathbb{C}^k$,  $F(\sum_{i=1}^n a_ib_i^t)=\sum_{i=1}^na_i\otimes b_i$. 
\item We say that $v\in\mathbb{C}^k\otimes\mathbb{C}^k$ is a Hermitian vector if $F^{-1}(v)\in M_k$ is Hermitian.

\end{enumerate}
\end{definition}

\begin{definition}\label{defS} Let $S:M_{k}\otimes M_k\rightarrow M_{k}\otimes M_k$ be defined by 
$$S(\sum_{i=1}^nA_i\otimes B_i)= \sum_{i=1}^nF(A_i)F(B_i)^t,$$  where $F(A_i)\in \mathbb{C}^k\otimes\mathbb{C}^k$ is a column vector and $F(B_i)^t$ is a row vector $($definition \ref{definition1}$)$. 
This map is usually called the ``realignment map" $($see \cite{chen},\cite{rudolph}$)$.
\end{definition}

\begin{remark} \label{spectral} Remind that $F$ is an isometry, i.e., $F(A)^t\overline{F(B)}=tr(AB^*)$, for every $A,B\in M_k$, where $F(A),F(B)\in\mathbb{C}^k\otimes\mathbb{C}^k$ and $\overline{F(B)}$ is the conjugation of the column  vector $F(B)$. We also have $tr(F^{-1}(v)F^{-1}(w)^*)=v^t\overline{w}$, for every $v,w\in\mathbb{C}^{k^2}$. 
  Therefore,  $A=\sum_{i=1}^n\lambda_i\gamma_i\otimes\overline{\gamma_i}$ is such that $\{\gamma_1,\ldots,\gamma_n\}$ is a orthonormal set of matrices of $M_k$  if and only if $S(A)=\sum_{i=1}^n\lambda_iv_i\overline{v}_i^t$ , where $F(\gamma_i)=v_i$ and $\{v_1,\ldots,v_n\}$ is an orthonormal set of eigenvectors of $S(A)$.
 \end{remark}

\begin{definition}\label{defmatricesinvariant} $($\textbf{Matrices Invariant under Realignment}$)$ Let $A\in M_k\otimes M_k$ be a positive semidefinite Hermitian matrix. We say that $A$ is invariant under realignment if  $A=S(A)$.
\end{definition}

\begin{examples}\label{example1} These examples shall be used in sections 5 and 6. 
 \begin{itemize}
 \item[a)] $Id\otimes Id+uu^t$ is invariant under realignment.\\
 By definitions \ref{definition1} and \ref{defS}, notice that $S(Id\otimes Id)=uu^t$. Now using property $(2)$ of \ref{propertiesofS}, notice that $S(uu^t)=Id\otimes Id$. Thus, $Id\otimes Id+uu^t$ is a positive semidefinite Hermitian matrix and $A=S(A)$.
 \item[b)] $Id\otimes Id+uu^t-T$ is invariant under realignment.\\
 Since $T=\sum_{i,j=1}^k e_ie_j^t\otimes e_je_i^t$ then $S(T)=T$. Now, the eigenvalues of $T$ are $1$ and $-1$ then $Id\otimes Id-T$ is positive semidefinite. Therefore $Id\otimes Id+uu^t-T$ is invariant under realignment.
 \end{itemize}
  
\end{examples}

\begin{lemma}\label{propertiesofS} $($Properties of the Realignment map$)$ Let $S:M_{k}\otimes M_k\rightarrow M_{k}\otimes M_k$ be the realignment map defined in \ref{defS}. Let $v_i,w_i\in\mathbb{C}^k\otimes\mathbb{C}^k$, $V,W,M,N\in M_k$.  Then

\begin{enumerate}
\item $S(\sum_{i=1}^nv_iw_i^t)=\sum_{i=1}^n F^{-1}(v_i)\otimes F^{-1}(w_i)$
\item $S^2=Id:M_{k}\otimes M_k\rightarrow M_{k}\otimes M_k$
\item $S((V\otimes W)A(M\otimes N))=(V\otimes M^t)S(A)(W^t\otimes N)$
\item $S(AT)T=A^{t_2}$ 
\item $S(A^{t_2})=S(A)T$ 
\item $S(AT)=S(A)^{t_2}$
\item $S(TAT)=S(A)^t$
\item $S(A^t)=TS(A)T$
\end{enumerate}  
\end{lemma}

\begin{proof} We only need to prove properties  $(3),(6),(7)$ and $(8)$, because $(1),(2)$ were proved in lemma 1.7 in \cite{carielloSPC} and $(4),(5)$ in lemma 2.1 in \cite{carielloSPC}. 

In order to prove properties $(3)$ and $(7)$, since both sides of the equation are linear on $A$, we just need to prove for $A=ab^t\otimes cd^t$, where $a,b,c,d\in\mathbb{C}^k$. 

Now, $S((V\otimes W)(ab^t\otimes cd^t)(M\otimes N))=S((Va\otimes Wc)(M^tb\otimes N^td)^t)$. By property $(1)$, this is equal to $F^{-1}(Va\otimes Wc)\otimes F^{-1}(M^tb\otimes N^td)=(Vac^tW^t)\otimes (M^t bd^tN)=(V\otimes M^t)(ac^t\otimes bd^t)(W^t\otimes N)=(V\otimes M^t)S(A)(W^t\otimes N)$. Thus, property $(3)$ is proved.

Now, $S(T(ab^t\otimes cd^t)T)=S(cd^t\otimes ab^t)=(c\otimes d)(a\otimes b)^t$ and $S((ab^t\otimes cd^t))^t=((a\otimes b)(c \otimes d)^t)^t=(c\otimes d)(a\otimes b)^t$. Thus, property $(7)$ is proved.

Next, by property $(5)$,   $S(S(A)^{t_2})=S(S(A))T$. By property $(2)$, we have  $S(S(A)^{t_2})=AT$. Again by property $(2)$, $S(AT)=S(S(S(A)^{t_2}))=S(A)^{t_2}$. Thus, we proved property $(6)$.

Finally, by properties $(2)$ and $(7)$ , $S(A^t)=S(S^2(A)^t)=S(S(TS(A)T))=TS(A)T$. Thus, we proved property $(8)$.
\end{proof}

\begin{lemma} \label{shape2} Let $A\in M_k\otimes M_k$ be a Hermitian matrix. If for every Hermitian vector $v\in\mathbb{C}^k\otimes\mathbb{C}^k$ $($definition \ref{definition1}$)$, we also obtain $Av \in\mathbb{C}^k\otimes\mathbb{C}^k$ Hermitian then $A$ has a Hermitian decomposition of the following type $\sum_{i=1}^n\alpha_i\gamma_i\otimes\gamma_i^t$, where $\alpha_i\in\mathbb{R}$.
\end{lemma}
\begin{proof}
Let $w\in\mathbb{C}^k\otimes\mathbb{C}^k$ be an eigenvector of $A$ associated to the eigenvalue $\lambda$. Let $w=w_1+iw_2$, where $w_1,w_2$ are Hermitian vectors. Since $A$ is a Hermitian matrix, $\lambda$ is a real number. Notice that $Aw=Aw_1+iAw_2=\lambda w_1+i \lambda w_2$. 

Now $Aw_1-\lambda w_1=i (\lambda w_2-Aw_2)$. Since $Aw_1-\lambda w_1$ and $\lambda w_2-Aw_2$ are Hermitian vectors, we obtain $0=Aw_1-\lambda w_1=\lambda w_2-Aw_2$.

Thus, every eigenvector of $A$ is a linear combination of Hermitian eigenvectors of $A$. Thus there is a set of Hermitian eigenvectors of $A$ that span a basis for $\mathbb{C}^k\otimes\mathbb{C}^k$ and we may extract a basis from this set. 
We can obtain an orthonormal basis of Hermitian eigenvectors. Therefore we obtain a spectral decomposition $A=\sum_{j} \alpha_jv_j\overline{v_j}^t$, where $\alpha_j$ are real numbers and $v_j$ Hermitian eigenvectors.

Next, $S(A)=S(\sum_{j=1}^n\alpha_jv_j\overline{v_j}^t)=\sum_{j=1}^n\alpha_j F^{-1}(v_j)\otimes \overline{F^{-1}(v_j)}$, by property $(1)$ in lemma \ref{propertiesofS}. Notice that $S(A)$ is a Hermitian matrix, since $\alpha_j\in\mathbb{R}$ and  $F^{-1}(v_j)$ is Hermitian for every $j$.

Thus, for every Hermitian vector $v\in\mathbb{C}^k\otimes\mathbb{C}^k$, $S(A)v$ is also Hermitian. Therefore, again we obtain a spectral decomposition $\sum_{j=1}^n\lambda_jw_j\overline{w_j}^t$ for $S(A)$, where $\lambda_j$ are real numbers and $w_j$ Hermitian eigenvectors.

Finally, by item 1 and 2 of lemma \ref{propertiesofS}, $A=S(S(A))=\sum_{j=1}^n \lambda_j F^{-1}(w_j)\otimes \overline{F^{-1}(w_j)}$, where $\lambda_j$ are real numbers and $F^{-1}(w_j)$ are Hermitian matrices. Thus,  $F^{-1}(w_j)^t= \overline{F^{-1}(w_j)}$. Define $\gamma_j=F^{-1}(w_j)$.
\end{proof}

\begin{lemma} \label{SPCequiv}Let $A\in M_k\otimes M_k$ be a positive semidefinite Hermitian matrix. $A$ is SPC if and only if $S(A^{t_2})$ is a positive semidefinite Hermitian matrix.
\end{lemma}
\begin{proof}
First, suppose $S(A^{t_2})$ is a positive semidefinite Hermitian matrix then  $S(A^{t_2})=\sum_{i=1}^n\lambda_i v_i\overline{v_i}^t$, where $\lambda_i>0$.  Therefore, $A^{t_2}=\sum_{i=1}^n \lambda_iV_i\otimes \overline{V_i}$, where $V_i=F^{-1}(v_i)$, by item 1 and 2 of lemma \ref{propertiesofS}.

Thus, $A^{t_2}$ is Hermitian and for every Hermitian vector $v\in\mathbb{C}^k\otimes\mathbb{C}^k$, we also have $A^{t_2}v$ Hermitian. Thus, by lemma \ref{shape2}, $A^{t_2}$ has a Hermitian decomposition $\sum_{i=1}^n\alpha_i\gamma_i\otimes\gamma_i^t$, where $\alpha_i\in\mathbb{R}$ and $\gamma_i$ is Hermitian for every $i$. Thus, $S(A^{t_2})=\sum_{i=1}^n\alpha_iw_i\overline{w_i}^t$, where $w_i=F(\gamma_i)$ and $\overline{w_i}=F(\overline{\gamma_i})=F(\gamma_i^t)$. Notice that since every $w_i$ is Hermitian then for every Hermitian $v\in\mathbb{C}^k\otimes\mathbb{C}^k$, $\overline{w_i}^tv=tr(\overline{F^{-1}(w_i)}^tF^{-1}(v))\in\mathbb{R}$ (see remark \ref{spectral}) and $S(A^{t_2})v$ is also Hermitian. Thus, $\mathbb{C}^k\otimes\mathbb{C}^k$ has an orthonormal basis of Hermitian eigenvectors of $S(A^{t_2})$ $($see proof of lemma \ref{shape2}$)$. Therefore, we can write  $S(A^{t_2})=\sum_{i=1}^n\beta_ir_i\overline{r_i}^t$, where $\beta_i$ are the positive eigenvalues of $S(A^{t_2})$ and $r_i$ are the orthonormal Hermitian eigenvectors.

Thus,  $A^{t_2}=\sum_{i=1}^n\beta_iR_i\otimes\overline{R_i}$,  where $F^{-1}(r_i)=R_i$, by property $(1)$ of lemma \ref{propertiesofS}. So $A=\sum_{i=1}^n\beta_iR_i\otimes R_i$ is a Hermitian Schmidt decomposition of $A$, since $F$ is an isometry and $R_i$ is Hermitian. Therefore $A$ is SPC. 

Now, if $A$ is SPC then $A$ is positive semidefinite with a Hermitian Schmidt decomposition $\sum_{i=1}^n\beta_iR_i\otimes R_i$, where $R_i$ is Hermitian and $\beta_i>0$ for every $i$. Now $S(A^{t_2})=\sum_{i=1}^n\beta_ir_i\overline{r_i}^t$ is a spectral decomposition, where $F(R_i)=r_i$, and $S(A^{t_2})$ is positive semidefinite. 
\end{proof}

\begin{lemma}\label{Formatinvariant} Let $A\in M_k\otimes M_k$ be a Matrix Invariant under Realignment. Then $A$ has a Hermitian Schimidt decomposition $\sum_{i=1}^n\lambda_i\gamma_i\otimes\gamma_i^t$ with $\lambda_i>0$, for every $i$.
\end{lemma}
\begin{proof} By definition \ref{defmatricesinvariant}, $A$ is a positive semidefinite Hermitian matrix. Thus, $A$ has a spectral decomposition $A=\sum_{i=1}^n\lambda_i v_i\overline{v_i}^t$, with $\lambda_i>0$ and $v_1,\ldots,v_n$ orthonormal.

Now, by property $(1)$ of \ref{propertiesofS}, $A=S(A)=\sum_{i=1}^n\lambda_i V_i\otimes\overline{V_i}$, where $V_i=F^{-1}(v_i)$. Therefore, $Av$ is Hermitian for every Hermitian $v\in\mathbb{C}^k\otimes\mathbb{C}^k$.

Therefore exists an orthonormal basis of $\mathbb{C}^k\otimes\mathbb{C}^k$ formed by Hermitian eigenvectors $($see proof of lemma \ref{shape2}$)$.
Thus, we may suppose without loss of generality that $v_1,\ldots,v_n$ above are Hermitian vectors. Thus, $V_1,\ldots, V_n$ are Hermitian and orthonormal $($since $F$ is an isometry$)$. Finally, define $\gamma_i=V_i$.
\end{proof}

\section{Completely Reducible Maps}

  Let us consider the usual inner product in $M_k:$ $\langle A,B\rangle=tr(AB^*)$.  In the context of positive maps, sometimes the term \textit{self-adjoint} means $L(A^*)=L(A)^*$ (see \cite{evans}). Here, we shall use this term with its usual meaning. We say that $L:VM_kV\rightarrow VM_kV$ is self-adjoint if $L$ is equal to its adjoint $L^*$ (i.e.$\langle L(A),B\rangle=\langle A,L(B)\rangle$).
 
 In this section, we use well known theorems of Perron-Frobenius Theory to describe some properties of completely reducible maps. These are theorems 2.3 and 2.5 in \cite{evans}: If $L:VM_kV\rightarrow VM_kV$ is a positive map then exists $\gamma\in P_k\cap VM_kV$ such that $L(\gamma)=\lambda\gamma$, where $\lambda$ is the spectral radius of $L$. Moreover, if $L$ is irreducible then this eigenvalue has multiplicity $1$.

Here we prove that if $L:VM_kV\rightarrow VM_kV$ is a self-adjoint positive map  then $L$ is completely reducible if and only if $L$ has the decomposition property (proposition \ref{propositioncompletelyreducible}). In the next section, we provide an equivalent way to prove that $L$ has the decomposition property (lemma \ref{lemmaequivalence}) and we shall give two applications of completely reducible maps in Quantum Information Theory. 
  
\begin{definition} \label{definitioncompletelyreducible} $($Completely Reducible Maps$)$: A positive map  $L:VM_kV\rightarrow VM_kV$ is called completely reducible, if there are orthogonal projections $V_1,\ldots,V_s\in M_k$ such that $V_iV_j=0\ (i\neq j)$, $V_iV=V_i$, $VM_kV=V_1M_kV_1\oplus\ldots \oplus V_sM_kV_s \oplus R$,  $R\perp V_1M_kV_1\oplus\ldots \oplus V_sM_kV_s$ and

\begin{enumerate}
\item $L(V_iM_kV_i)\subset V_iM_kV_i$,
\item $L|_{V_iM_kV_i}$ is irreducible,
\item $L|_R\equiv 0$.
\end{enumerate}
\end{definition}  
  
\begin{definition}: \label{definitiondecompositionproperty} Let $L:VM_kV\rightarrow VM_kV$ be a self-adjoint positive map.  We say that $L$ has the decomposition property if for every  $\gamma\in P_k\cap VM_kV$ such that $L(\gamma)=\lambda\gamma$, $\lambda>0$ and $V_1\in M_k$ is the orthogonal projection onto $\Im(\gamma)$ then $L|_{R}\equiv 0$, where $R=(V-V_1)M_kV_1\oplus V_1M_k(V-V_1)$. Notice that $R$ is the orthogonal complement of $V_1M_kV_1\oplus (V-V_1)M_k(V-V_1)$ in $VM_kV$.
\end{definition}

The next lemma is well known.

\begin{lemma}\label{lemmawellknown} Let $L:VM_kV\rightarrow WM_mW$ be a positive map. If $\gamma\in P_k\cap VM_kV$ and $L(\gamma)=\delta$ then $L(V_1M_kV_1)\subset W_1M_mW_1$, where $V_1$ is the orthogonal projection onto $\Im(\gamma)$ and $W_1$ is the orthogonal projection onto $\Im(\delta)$. 
\end{lemma}

\begin{corollary}\label{corollarywellknown} Let $L:VM_kV\rightarrow VM_kV$ be a positive map and $\gamma\in P_k\cap VM_kV$ be such that $L(\gamma)=\lambda\gamma,$ $\lambda>0$. Then, $L(V_1M_kV_1)\subset V_1M_kV_1$, where $V_1$ is the orthogonal projection onto $\Im(\gamma)$.
\end{corollary}

\begin{lemma} \label{lemmairreducible} Let $L:VM_kV\rightarrow VM_kV$ be a self-adjoint positive map. $L$ is irreducible if and only if the biggest eigenvalue has multiplicity $1$ with respect to an eigenvector $\gamma\in P_k\cap VM_kV$ such that $\Im(\gamma)=\Im(V)$.
\end{lemma}
\begin{proof}
Since $L$ is self-adjoint, the eigenvalues of $L$ are real numbers. Since $L:VM_kV\rightarrow VM_kV$ is a positive map, by theorem 2.5 in \cite{evans}, the spectral radius $\lambda$ is an eigenvalue and exists $\gamma\in P_k\cap VM_kV$ such that $L(\gamma)=\lambda\gamma$. Therefore the spectral radius is the biggest eigenvalue of $L$. Since $L$  is irreducible, the multiplicity of $\lambda$ is $1$ by proposition 2.3 in \cite{evans}. Let $V_1\in M_k$ be the orthogonal projection onto $\Im(\gamma)$. Notice that $\Im(V_1)\subset\Im(V)$. By the previous corollary $L(V_1M_kV_1)\subset V_1M_kV_1$. Since $L$ is irreducible then $V_1=V$ and $\Im(\gamma)=\Im(V)$.

Now, for the converse if $L(V_1M_kV_1)\subset V_1M_kV_1$, $\Im(V_1)\subset\Im(V)$, then the positive map $L:V_1M_kV_1\rightarrow V_1M_kV_1$ has an eigenvector $\gamma'\in  P_k\cap V_1M_kV_1$, by proposition 2.5 in \cite{evans}. If  $\Im(V_1)\neq\Im(V)$ then $\Im(\gamma')\neq\Im(\gamma)$ and $\gamma'$ is not a multiple of $\gamma$. Since the multiplicity of the biggest eigenvalue is $1$ then $\gamma'$ is associated to a different eigenvalue. Thus, $\gamma'$ is orthogonal to $\gamma$, since $L$ is self-adjoint. However, $\gamma'$ and $\gamma$ are positive semidefinite and $\Im(\gamma')\subset\Im(V_1)\subset\Im(V)=\Im(\gamma)$, thus they can not be orthogonal. Thus, $\Im(V_1)=\Im(V)$ and $V_1=V$. Therefore, $L$ is irreducible. 
\end{proof}

\begin{lemma}\label{lemmareduction} Let $L:VM_kV\rightarrow VM_kV$ be a self-adjoint positive map. Let us assume that $L$ has the decomposition property (definition \ref{definitiondecompositionproperty}). Let $V'M_kV'\subset VM_kV$ be such that $L(V'M_kV')\subset V'M_kV'$  then $L|_{V'M_kV'}$ has also the decomposition property. 
\end{lemma}
\begin{proof} Let $\gamma\in  P_k\cap V'M_kV'$ be such that $L(\gamma)=\lambda\gamma, \lambda>0$. Since $L:VM_kV\rightarrow VM_kV$ has the decomposition property (definition \ref{definitiondecompositionproperty}) then $L|_{R}\equiv 0$, where $R=(V-V_1)M_kV_1\oplus V_1M_k(V-V_1)$  and $V_1\in M_k$ is the orthogonal projection such that $\Im(V_1)=\Im(\gamma)$. Notice that $\Im(V_1)=\Im(\gamma)\subset\Im(V')\subset\Im(V)$.

Now, consider $R'=(V'-V_1)M_kV_1\oplus V_1M_k(V'-V_1)$.
Since  $(V'-V_1)M_kV_1=(V-V_1)(V'-V_1)M_kV_1\subset (V-V_1)M_kV_1$ and $V_1M_k(V'-V_1)=V_1M_k(V'-V_1)(V-V_1)\subset V_1M_k(V-V_1)$ then $R'\subset R$ and $L|_{R'}\equiv 0$.
Thus, $L:V'M_kV'\rightarrow V'M_kV'$ has the decomposition property.
\end{proof}

\begin{proposition} \label{propositioncompletelyreducible} If $L:VM_kV\rightarrow VM_kV$ is a self-adjoint positive map then $L$ has the decomposition property if and only if $L$ is completely reducible. Moreover, the orthogonal projections $V_1,\ldots,V_s$ in definition \ref{definitioncompletelyreducible} are unique and $s\geq$ the multiplicity of the biggest eigenvalue of $L$.
\end{proposition}

\begin{proof} First, suppose that $L$ has the decomposition property and let us prove that $L$ is completely reducible by induction on the rank of $V$. Notice that if $\text{rank}(V)=1$ then $\dim(VM_kV)=1$ and $L$ is irreducible on $VM_kV$. Thus, $L$ is completely reducible by definition \ref{definitioncompletelyreducible}. Let us assume that $\text{rank}(V)>1$.

Since $L$ is a positive map then  $S=\{\gamma,\ 0\neq\gamma\in  P_k\cap VM_kV,\ L(\gamma)=\lambda\gamma, \lambda>0 \}\neq \emptyset $, by proposition 2.5 in \cite{evans}. Let $\gamma\in S$ be such that $\text{rank}(\gamma)=\min\{\text{rank}(\gamma'),\ \gamma' \in S\}$.

By corollary \ref{corollarywellknown}, $L(V_1M_kV_1)\subset V_1M_kV_1$, where $V_1$ is the orthogonal projection onto $\Im(\gamma)$. 
Now, if  $L|_{V_1M_kV_1}$ is not irreducible then exists $V_1'M_kV_1'\subset V_1M_kV_1$ with $\text{rank}(V_1')<\text{rank}(V_1)$ and $L(V_1'M_kV_1')\subset V_1'M_kV_1'$.

By proposition 2.5 in \cite{evans}, exists $0\neq\delta\in  P_k\cap V_1'M_kV_1'$ such that $L(\delta)=\mu\delta$, $\mu>0$. However, $\text{rank}(\delta)\leq \text{rank}(V_1')<\text{rank}(V_1)=\text{rank}(\gamma)$. This is a  contradiction with the choice of $\gamma$.
Thus, $L|_{V_1M_kV_1}$ is irreducible. 

Now, if $\text{rank}(V_1)=\text{rank}(V)$ then $V_1=V$ and $L|_{VM_kV}$ is irreducible. Therefore, $L:VM_kV\rightarrow VM_kV$ is completely reducible by definition \ref{definitioncompletelyreducible}. 

Next, suppose $\text{rank}(V_1)<\text{rank}(V)$.
Since $L(V_1M_kV_1)\subset V_1M_kV_1$ and $L$ is self-adjoint then $L((V_1M_kV_1)^{\perp})\subset (V_1M_kV_1)^{\perp}$. Therefore, $tr(L(V-V_1) V_1)=0$. Since $L(V-V_1)$ and $V_1$ are positive semidefinite then $\Im(L(V-V_1))\subset\Im(V-V_1)$. By lemma \ref{lemmawellknown}, $L((V-V_1)M_k(V-V_1))\subset (V-V_1)M_k(V-V_1)$.

Notice that $L|_{(V-V_1)M_k(V-V_1)}$ is a self adjoint positive map with the decomposition property by lemma \ref{lemmareduction}. Since $\text{rank}(V-V_1)<\text{rank}(V)$, by induction on the rank, $L|_{(V-V_1)M_k(V-V_1)}$ is completely reducible.

Thus, there are orthogonal projections $V_2,\ldots,V_s\in M_k$ satisfying $V_iV_j= 0\ (i\neq j)$, $V_i(V-V_1)=V_i\ (i\geq2)$,  $(V-V_1)M_k(V-V_1)=V_2M_kV_2\oplus\ldots \oplus V_sM_kV_s \oplus \widetilde{R}$ with 
$\widetilde{R}\perp V_2M_kV_2\oplus\ldots \oplus V_sM_kV_s$, $L|_{V_iM_kV_i}$ is irreducible for $2\leq i\leq s$ and $L|_{\widetilde{R}}\equiv 0$. 

Now, since $L$  has the decomposition property then $VM_kV=V_1M_kV_1\oplus (V-V_1)M_k(V-V_1) \oplus R$, where  $L|_{R}\equiv 0$ and $R\perp V_1M_kV_1\oplus (V-V_1)M_k(V-V_1)$. 

Thus, we obtained $VM_kV=V_1M_kV_1\oplus V_2M_kV_2\oplus\ldots \oplus V_sM_kV_s \oplus \widetilde{R} \oplus R$ such that  $L|_{V_iM_kV_i}$ is irreducible for $1\leq i\leq s$ and $L|_{\widetilde{R}\oplus R}\equiv 0$. Notice that $V_iV_j=0$, for  $2\leq i\neq j\leq s$ and $V_1V_i=0$, for $2\leq i\leq s$, because $\Im(V_i)\subset\Im(V-V_1)$.

Notice that $\widetilde{R}\perp  V_2M_kV_2\oplus\ldots \oplus V_sM_kV_s$ and $\widetilde{R}\perp V_1M_kV_1$, because $\widetilde{R}\subset (V-V_1)M_k(V-V_1)$. Therefore  $\widetilde{R}\oplus R\perp V_1M_kV_1\oplus V_2M_kV_2\oplus\ldots \oplus V_sM_kV_s$ and $L|_{\widetilde{R}\oplus R}\equiv 0$. Thus, $L$ is completely reducible.

For the converse let us assume that $L$ is completely reducible and let us prove that $L$ has the decomposition property. Thus, $VM_kV=V_1M_kV_1\oplus\ldots \oplus V_sM_kV_s \oplus R$,  $R\perp V_1M_kV_1\oplus\ldots \oplus V_sM_kV_s$, $L(V_iM_kV_i)\subset V_iM_kV_i$, $L|_{V_iM_kV_i}$ is irreducible and $L|_R\equiv 0$.

Assume $L(\gamma')=\lambda\gamma'$, $\lambda>0$ and $\gamma'\in P_k\cap VM_kV$ and let $V'\in M_k$ be the orthogonal projection onto $\Im(\gamma')$. By corollary \ref{corollarywellknown}, we have $L(V'M_kV')\subset V'M_kV'$.

Notice that, $\gamma'=\gamma_1'+\ldots+\gamma_s'$, where $\gamma_i'\in V_iM_kV_i$. Now, since $\Im(\gamma_i')\subset\Im(V_i)$ and $\Im(V_i)\perp\Im(V_j)$, for $i\neq j$, then each $\gamma_i'\in P_k$. Since each $V_iM_kV_i$ is an invariant subspace of $L$ then we must also conclude that $L(\gamma_i')=\lambda\gamma_i'$. Notice that not all $\gamma_i'=0$. Assume without loss of generality that $\gamma'=\gamma_1'+\ldots+\gamma_m'$ and $\gamma_i'\neq 0$, for  $1\leq i\leq m\leq s$.   

Now, if for some $1\leq i\leq m$, $\Im(\gamma_i')\neq \Im(V_i)$  then $L|_{V_iM_kV_i}$ is not irreducible, by corollary \ref{corollarywellknown}, which is a contradiction. Therefore, $\Im(\gamma_i')=\Im(V_i)$ for $1\leq i\leq m$ and $V_1+\dots+V_m=V'$.
    
Next, $VM_kV=V'M_kV'\oplus (V-V')M_k(V-V')\oplus R'$, where $R'=(V-V')M_kV'\oplus V'M_k(V-V')$. Notice that $R'\perp V'M_kV'\oplus (V-V')M_k(V-V')$.  

Now, $V_1M_kV_1\oplus\ldots \oplus V_mM_kV_m\subset V'M_kV'$ and $V_{m+1}M_kV_{m+1}\oplus\ldots \oplus V_sM_kV_s\subset (V-V')M_k(V-V')$, therefore $R'\perp V_1M_kV_1\oplus\ldots \oplus V_sM_kV_s$ and $R'\subset R$. Therefore, $L|_{R'}\equiv 0$ and $L$ has the decomposition property by definition \ref{definitiondecompositionproperty}.

Finally, if $L:VM_kV\rightarrow VM_kV$ is a self adjoint completely reducible map then the non-null eigenvalues of $L$ are the non-null eigenvalues of $L|_{V_iM_kV_i}$. Since $L|_{V_iM_kV_i}$  is irreducible then the
multiplicity of the biggest eigenvalue is $1$ by lemma \ref{lemmairreducible}. Therefore each $L|_{V_iM_kV_i}$ has at most one biggest eigenvalue of $L$. Thus, $s\geq$ the multiplicity of the biggest eigenvalue of $L:VM_kV\rightarrow VM_kV$. Now, if $L(V''M_kV'')\subset V''M_kV''$ and $L|_{V''M_kV''}$ is irreducible then by lemma \ref{lemmairreducible}, there is $\gamma''\in P_k\cap V''M_kV''$ such that 
$L(\gamma'')=\lambda\gamma''$, $\lambda>0$ and $\Im(\gamma'')=\Im(V'')$.
As we noticed in the second part of this proof, 
there is  $V_{i}M_kV_{i}\subset V''M_kV''$ ($V''$ is a sum of some $V_i$'s). Since   $L(V_{i}M_kV_{i})\subset V_iM_kV_{i}$ then $L|_{V''M_kV''}$ is irreducible if and only if  $V''=V_i$, for some $1\leq i\leq s$.
\end{proof}

\section{Two Applications in Quantum Information Theory}

Let $A\in M_k\otimes M_m\simeq M_{km}$, $A\in P_{km}$. Thus, $A$ has a Hermitian decomposition $\sum_{i=1}^n A_i\otimes B_i$ (see \cite{cariello}). Let  $F_A: M_m\rightarrow M_k$ be $F_{A}(X)=\sum_{i=1}^ntr(B_iX)A_i$ and $G_A: M_k\rightarrow M_m$ be $G_{A}(X)=\sum_{i=1}^ntr(A_iX)B_i$.  
 These maps are adjoints with respect to the trace inner product (Since $A_i,B_i$ are Hermitian matrices then $F_A(Y^*)=F_A(Y)^*$, for every $Y\in M_m$. Notice that if $X\in M_k$ and $Y\in M_m$ then $tr(A(X \otimes Y^*))=tr(G_A(X)Y^*)=tr(XF_A(Y^*))=tr(XF_A(Y)^*)$). 
 
 These maps are also positive maps and $F_A\circ G_A:M_k\rightarrow M_k$ is a self-adjoint positive map. Notice that if $\{\gamma_1,\ldots,\gamma_{k^2}\}$ is an orthonormal basis of $M_k$ formed by Hermitian matrices then $A=\sum_{i=1}^{k^2}\gamma_i\otimes G_A(\gamma_i) $.

In this section, we assume that $F_A\circ G_A:M_k\rightarrow M_k$ is completely reducible. This is a very strong restriction. However, in the next section, we shall prove that if $A$ is PPT or SPC or invariant under realignment then $F_A\circ G_A:M_k\rightarrow M_k$ is indeed completely reducible.  We begin this section with a simple lemma that provides an equivalent way to prove that $F_A\circ G_A:M_k\rightarrow M_k$ is completely reducible. This lemma shall be used in the next section in order to prove these theorems. 

Here, we assume that $F_A\circ G_A:M_k\rightarrow M_k$ is completely reducible and we give two applications in Quantum Information Theory. The first application is the reduction of the separability problem to the weakly irreducible case (corollary \ref{corollarysplitdecomposition}) and the second is proposition \ref{propositioneigenvalues1} which grants the separability of $A$, if $F_A\circ G_A:M_k\rightarrow M_k$ has only eigenvalues 1 or 0. These theorems were proved in \cite{cariello}, when $A$ is PPT or SPC. Thus, we extend this theorems for matrices invariant under realignment.

In section 5, we present our last application concerning mutually unbiased bases using this proposition \ref{propositioneigenvalues1} for a matrix invariant under realignment (see proposition \ref{mainapplication} and theorem \ref{thelastbasis}).

\begin{lemma} \label{lemmaequivalence} Let $A\in M_k\otimes M_m\simeq M_{km}$, $A\in P_{km}$.  Thus, $F_A\circ G_A:M_k\rightarrow M_k$ is completely reducible if and only if for every  $\gamma\in P_{k}$ such that $F_A\circ G_A(\gamma)=\lambda\gamma$, $\lambda>0$, we have $A=(V_1\otimes W_1)A(V_1\otimes W_1)+(Id-V_1\otimes Id-W_1)A(Id-V_1\otimes Id-W_1),$ where $V_1\in M_k, W_1\in M_m$ are orthogonal projections onto $\Im(\gamma),\Im(G_A(\gamma))$, respectively.
\end{lemma}
\begin{proof}

Suppose $F_A\circ G_A:M_k\rightarrow M_k$ is completely reducible then
 $F_A\circ G_A:M_k\rightarrow M_k$ has the decomposition property (definition \ref{definitiondecompositionproperty}) by proposition \ref{propositioncompletelyreducible}. 

If $\gamma\in P_{k}$ is such that $F_A\circ G_A(\gamma)=\lambda\gamma$, $\lambda>0$, then $M_k=V_1M_kV_1\oplus (Id-V_1)M_k(Id-V_1)\oplus R$, where $R\perp V_1M_kV_1\oplus (Id-V_1)M_k(Id-V_1)$ and $F_A\circ G_A|_{R}\equiv 0$, where $V_1\in M_k$ is the orthogonal projection onto $\Im(\gamma)$.

Next, let $W_1\in M_m$ be the orthogonal projection onto the $\Im(G_A(\gamma))$. By lemma \ref{lemmawellknown}, we have $G_A(V_1M_kV_1)\subset W_1M_mW_1$, because $G_A$ is a positive map, since $A\in P_{km}$.

Now, $\langle G_A(Id-V_1), G_A(\gamma)\rangle=\langle Id-V_1, F_A\circ G_A(\gamma)\rangle=\lambda \langle Id-V_1, \gamma\rangle=0$. Since $G_A(Id-V_1)$ and $G_A(\gamma)$ are postive semidefinite then $\Im(G_A(Id-V_1))\perp\Im(G_A(\gamma))=\Im(W_1)$. Thus, $\Im(G_A(Id-V_1))\subset \Im(Id-W_1)$. Again by lemma \ref{lemmawellknown}, we have $G_A((Id-V_1)M_k(Id-V_1))\subset (Id-W_1)M_m(Id-W_1)$.

Now since $F_A\circ G_A|_{R}\equiv 0$ and $F_A$,$G_A$ are adjoint maps then $G_A|_{R}\equiv 0$.

Now, let $\{\gamma_1,\ldots,\gamma_r\}$ be an orthonormal basis of $V_1M_kV_1$ formed by Hermitian matrices, $\{\delta_1,\ldots,\delta_s\}$ be an orthonormal basis of $(Id-V_1)M_k(Id-V_1)$ formed by Hermitian matrices and $\{\alpha_1,\ldots,\alpha_t\}$ be an orthonormal basis of $R$ formed by Hermitian matrices. Then $A=\sum_{i=1}^r\gamma_i\otimes G_A(\gamma_i)+\sum_{i=1}^s\delta_i\otimes G_A(\delta_i)+\sum_{i=1}^s\alpha_i\otimes G_A(\alpha_i)$.
Since $G_A(\alpha_i)=0$ then $A=\sum_{i=1}^r\gamma_i\otimes G_A(\gamma_i)+\sum_{i=1}^s\delta_i\otimes G_A(\delta_i)$.

Now, since $\{\gamma_1,\ldots,\gamma_r\}\subset V_1M_kV_1$ then $G_A(\gamma_i)\in W_1M_mW_1$ and since 
 $\{\delta_1,\ldots,\delta_s\}\subset (Id-V_1)M_k(Id-V_1)$ then  $G_A(\delta_i)\in (Id-W_1)M_m(Id-W_1)$. Therefore, 
 $(V_1\otimes W_1)A(V_1\otimes W_1)=\sum_{i=1}^r\gamma_i\otimes G_A(\gamma_i)$ and $(Id-V_1\otimes Id-W_1)A(Id-V_1\otimes Id-W_1)=\sum_{i=1}^s\delta_i\otimes G_A(\delta_i)$ and $A=(V_1\otimes W_1)A(V_1\otimes W_1)+(Id-V_1\otimes Id-W_1)A(Id-V_1\otimes Id-W_1)$.
 
Now, for the converse, assume that if $\gamma\in P_{k}$ is such that $F_A\circ G_A(\gamma)=\lambda\gamma$, $\lambda>0$ then $A=(V_1\otimes W_1)A(V_1\otimes W_1)+(Id-V_1\otimes Id-W_1)A(Id-V_1\otimes Id-W_1),$ where $V_1, W_1$ are orthogonal projections onto $\Im(\gamma),\Im(G_A(\gamma))$, respectively.

Let $M_k=V_1M_kV_1\oplus (Id-V_1)M_k(Id-V_1)\oplus R$, $R\perp V_1M_kV_1\oplus (Id-V_1)M_k(Id-V_1)$.

Notice that $G_A|_{R}\equiv 0$ and $F_A\circ G_A|_{R}\equiv 0$. Therefore, $F_A\circ G_A:M_k\rightarrow M_k$ has the decomposition property (definition \ref{definitiondecompositionproperty}) and by proposition \ref{propositioncompletelyreducible}, $F_A\circ G_A:M_k\rightarrow M_k$ is completely reducible.
\end{proof}

\begin{definition}\label{definitionweaklyirreducible}  Let $A\in M_k\otimes M_m\simeq M_{km}$, $A\in P_{km}$. We say that $A$ is weakly irreducible if for every orthogonal projections $V_1,V_2\in M_k$ and $W_1,W_2\in M_m$ such that $V_2=Id-V_1$, $W_2=Id-W_1$ and $A=(V_1\otimes W_1) A(V_1\otimes W_1)+(V_2\otimes W_2) A(V_2\otimes W_2)$, we obtain $(V_1\otimes W_1) A(V_1\otimes W_1)=0$ or $(V_2\otimes W_2) A(V_2\otimes W_2)=0$.
\end{definition}

\begin{proposition}\label{propositionweaklyirred}  Let $A\in M_k\otimes M_m\simeq M_{km}$, $A\in P_{km}$. Let $\sum_{i=1}^n\lambda_i\gamma_i\otimes\delta_i$ be a Hermitian Schmidt decomposition of $A$ such that $\lambda_1\geq\lambda_2\geq\ldots\geq\lambda_n>0$. If $F_A\circ G_A:M_k\rightarrow M_k$ is completely reducible  then the following conditions are equivalent:
\begin{enumerate}
\item $A$ is weakly irreducible,
\item $s=1$ in definition \ref{definitioncompletelyreducible} with $L=F_A\circ G_A:M_k\rightarrow M_k$,
\item $\lambda_1>\lambda_2$ and $\Im(\gamma_i)\subset\Im(\gamma_1)$, $\Im(\delta_i)\subset\Im(\delta_1)$, for $1\leq i\leq n$.
\end{enumerate} 
\end{proposition}
\begin{proof}
 Notice that $F_A\circ G_A(\gamma_i)=\lambda_i^2\gamma_i$ for $1\leq i \leq n$. Thus, the biggest eigenvalue of $F_A\circ G_A$ is $\lambda_1^2$. By definition \ref{definitioncompletelyreducible}, $M_k=V_1M_kV_1\oplus\ldots \oplus V_sM_kV_s \oplus R$, $F_A\circ G_A|_{V_iM_kV_i}$ is irreducible and $F_A\circ G_A|_{R}\equiv 0$, where $R\perp V_1M_kV_1\oplus\ldots \oplus V_sM_kV_s$. Since the eigenvalues of $F_A\circ G_A:M_k\rightarrow M_k$ are the eigenvalues of $F_A\circ G_A|_{V_iM_kV_i}$, for $1\leq i\leq s$, then $\lambda_1^2$ is the biggest eigenvalue of some $F_A\circ G_A|_{V_iM_kV_i}$. Without loss of generality we may assume that 
exists $0\neq\gamma\in P_{k}\cap V_1M_kV_1$ such that $F_A\circ G_A(\gamma)=\lambda_1^2\gamma$ and $\Im(\gamma)=\Im(V_1)$, by lemma \ref{lemmairreducible}. Thus, by lemma \ref{lemmaequivalence}, $A=(V_1\otimes W_1)A(V_1\otimes W_1)+(Id-V_1\otimes Id-W_1)A(Id-V_1\otimes Id-W_1),$ where $V_1, W_1$ are orthogonal projections onto $\Im(\gamma),\Im(G_A(\gamma))$, respectively.

First, let us assume that $A$ is weakly irreducible, then or $(V_1\otimes W_1)A(V_1\otimes W_1)=0$ or $(Id-V_1\otimes Id-W_1)A(Id-V_1\otimes Id-W_1)=0$. Notice that if $(V_1\otimes W_1)A(V_1\otimes W_1)=0$ then $A=(Id-V_1\otimes Id-W_1)A(Id-V_1\otimes Id-W_1)$ and $G_A(\gamma)=0$, since $\gamma\in V_1M_kV_1$. Therefore,  $0=F_A\circ G_A(\gamma)=\lambda_1^2\gamma$, which is a contradiction. Therefore $(Id-V_1\otimes Id-W_1)A(Id-V_1\otimes Id-W_1)=0$ and $A=(V_1\otimes W_1)A(V_1\otimes W_1)$. In this case, $G_A|_{(V_1M_k V_1)^{\perp}}\equiv 0$ and $F_A\circ G_A|_{(V_1M_k V_1)^{\perp}}\equiv 0$. Thus, $s=1$ in definition \ref{definitioncompletelyreducible}.

Second, suppose that $s=1$ in definition \ref{definitioncompletelyreducible} then $M_k=V_1M_k V_1\oplus R$, $F_A\circ G_A|_{V_1M_k V_1}$ is irreducible and $F_A\circ G_A|_{R}\equiv 0$, where $R=(V_1M_k V_1)^{\perp}$.
Thus,  $\gamma_i\in V_1M_kV_1$ for $1\leq i\leq n$, since  $F_A\circ G_A(\gamma_i)=\lambda_i^2\gamma_i$ and $F_A\circ G_A(M_k)=F_A\circ G_A(V_1M_kV_1)\subset V_1M_kV_1$.
By lemma \ref{lemmawellknown}, $G_A(V_1M_kV_1)\subset W_1M_mW_1$, since $\Im(G_A(\gamma))=\Im(W_1)$. Thus, $\lambda_i\delta_i=G_A(\gamma_i)\in W_1M_mW_1$ and $\Im(\delta_i)\subset\Im(W_1)$.

Next, since $F_A\circ G_A:V_1M_kV_1\rightarrow V_1M_kV_1$ is irreducible then the multiplicity of the biggest eigenvalue is $1$ by lemma \ref{lemmairreducible}, thus $\lambda_1^2>\lambda_2^2$ and $\lambda_1>\lambda_2$. Moreover, $\gamma$ must be a multiple of $\gamma_1$, because $F_A\circ G_A(\gamma_1)=\lambda_1^2\gamma_1$.
 
Thus, $G_A(\gamma)$ is also a multiple of $\delta_1$. Therefore, $\Im(\gamma_i)\subset\Im(V_1)=\Im(\gamma)=\Im(\gamma_1)$ and $\Im(\delta_i)\subset\Im(W_1)=\Im(G_A(\gamma))=\Im(\delta_1)$.

Finally, if $\lambda_1>\lambda_2$ and $\Im(\gamma_i)\subset\Im(\gamma_1)$, $\Im(\delta_i)\subset\Im(\delta_1)$, for $1\leq i\leq n$ then $A$ is weakly irreducible by theorem 44 in \cite{cariello}.
\end{proof}

\begin{proposition} \label{splitdecomposition} Let $A\in M_k\otimes M_m\simeq M_{km}$, $A\in P_{km}$. If $F_A\circ G_A:M_k\rightarrow M_k$ is completely reducible  then $A=\sum_{i=1}^s(V_i\otimes W_i)A(V_i\otimes W_i)$ such that 
\begin{enumerate}
\item $V_1,\ldots,V_s\in M_k$ are orthogonal projections such that $V_iV_j=0$
\item $W_1,\ldots,W_s\in M_m$ are orthogonal projections such that $W_iW_j=0$
\item $(V_i\otimes W_i)A(V_i\otimes W_i)$ is weakly irreducible and non-null for every $i$.
\item $s\geq$ multiplicity of the biggest eigenvalue of $F_A\circ G_A:M_k\rightarrow M_k$.
\end{enumerate}
\end{proposition}
\begin{proof}

Since $F_A\circ G_A:M_k\rightarrow M_k$ is completely reducible then $M_k=V_1M_kV_1\oplus\ldots \oplus V_sM_kV_s \oplus R$, $F_A\circ G_A(V_iM_kV_i)\subset V_iM_kV_i$, $F_A\circ G_A|_{V_iM_kV_i}$ is irreducible, $F_A\circ G_A|_{R}\equiv 0$ and $s\geq$ multiplicity of the biggest eigenvalue of $F_A\circ G_A:M_k\rightarrow M_k$, by proposition \ref{propositioncompletelyreducible}. 

By lemma \ref{lemmairreducible}, there is $\gamma_1^j\in P_k\cap V_jM_kV_j$, $1\leq j\leq s$, such that $\gamma_1^j$ is an eigenvector of $F_A\circ G_A:V_jM_kV_j\rightarrow V_jM_kV_j$ associated to the unique biggest eigenvalue and $\Im(\gamma_1^j)=\Im(V_j)$. Since $G_A$ is a positive map then $G_A(\gamma_1^j)\in P_m$.

By lemma \ref{lemmawellknown}, $G_A(V_jM_kV_j)\subset W_jM_kW_j$, where $W_j$ is the orthogonal projection onto $\Im(G_A(\gamma_1^j))$. Notice that $V_jM_kV_j\perp V_iM_kV_i$, for $i\neq j$, since $V_iV_j=0$. Therefore, $\langle G_A(\gamma_1^j),G_A(\gamma_1^i)\rangle=\langle \gamma_1^j,F_A\circ G_A(\gamma_1^i)\rangle=0$, for $i\neq j$. Thus, $W_iW_j=0$ for $i\neq j$.

Let $\{\gamma_1^j,\ldots,\gamma_{r_j}^j\}$ be an orthonormal basis of $V_jM_kV_j$ formed by Hermitian matrices. 
Let $\{\delta_1,\dots,\delta_r\}$  be an orthonormal basis of $R$  formed by Hermitian matrices.  Thus, 
$\bigcup_{j=1}^s\{\gamma_1^j,\ldots,\gamma_{r_j}^j\}\cup\{\delta_1,\dots,\delta_r\}$ is an orthonormal basis of $M_k$ formed by Hermitian matrices.
Let $A_j=\gamma_1^j\otimes G_A(\gamma_1^j)+\ldots+\gamma_{r_j}^j\otimes G_A(\gamma_{r_j}^j)$. 
Thus, $A=\sum_{j=1}^s A_j+\delta_1\otimes G_A(\delta_1)+\dots+\delta_r\otimes G_A(\delta_r)$

Now, since $F_A\circ G_A|_{R}\equiv 0$ and $F_A, G_A$ are adjoints then $G_A|_{R}\equiv 0$ and $A=\sum_{j=1}^s A_j$.

Next, $(V_j\otimes W_j)A(V_j\otimes W_j)=A_j$, since $\gamma_l^j\in V_jM_kV_j$, $G_A(\gamma_l^j)\in W_jM_mW_j$, $V_iV_j=0$ and $W_iW_j=0$ for $i\neq j$. Therefore, $A_j\in P_{km}$.

Notice that, $F_{A_j}\circ G_{A_j}|_{V_jM_kV_j}=F_A\circ G_A|_{V_jM_kV_j}$ which is irreducible. Therefore $A_j\neq 0$ for $1\leq j\leq s$. Next, $M_k=V_jM_kV_j\oplus (V_jM_kV_j)^{\perp}$ and $F_{A_j}\circ G_{A_j}((V_jM_kV_j)^{\perp})=0$.  Therefore, by definition \ref{definitioncompletelyreducible}, $F_{A_j}\circ G_{A_j}: M_k\rightarrow M_k$ is completely reducible with $s=1$ . Finally, by item 2 of proposition \ref{propositionweaklyirred}, $A_j$ is weakly irreducible.
\end{proof}

\begin{corollary}\label{corollarysplitdecomposition} Let $A$ be the matrix of proposition \ref{splitdecomposition}. Then $A$ is separable if and only if each $(V_i\otimes W_i)A(V_i\otimes W_i)$ is separable. Thus, for this type of $A$ the separability problem is reduced to the weakly irreducible case.
\end{corollary}

\begin{proposition}\label{propositioneigenvalues1}  Let $A\in M_k\otimes M_m\simeq M_{km}$, $A\in P_{km}$. If $F_A\circ G_A:M_k\rightarrow M_k$ is completely reducible with all eigenvalues equal to $1$ or $0$ then exists a unique Hermitian Schmidt decomposition of $A$, $\sum_{i=1}^n\gamma_i\otimes\delta_i$, such that $\gamma_i\in P_{k},\delta_i\in P_{m}$. Therefore, $A$ is separable.
\end{proposition}
\begin{proof} Suppose the multiplicity of the eigenvalue $1$ is $n$. Since $F_A\circ G_A:M_k\rightarrow M_k$ is completely reducible then there are orthogonal projections $V_1,\ldots,V_s$ such that $\Im(V_i)\perp\Im(V_j)$,
 $M_k=V_1M_kV_1\oplus\ldots \oplus V_sM_kV_s \oplus R$, $F_A\circ G_A(V_iM_kV_i)\subset V_iM_kV_i$, $F_A\circ G_A|_{V_iM_kV_i}$ is irreducible, $F_A\circ G_A|_{R}\equiv 0$ and $s\geq n$, by definition \ref{definitioncompletelyreducible} and proposition \ref{propositioncompletelyreducible}. Remind that 
each $F_A\circ G_A|_{V_iM_kV_i}$ has a unique biggest eigenvalue, since $F_A\circ G_A|_{V_iM_kV_i}$ is irreducible by lemma \ref{lemmairreducible}. Moreover, the eigenvalues of $F_A\circ G_A|_{V_iM_kV_i}$ are $1$ or $0$. Thus, $s=n$ and for each $F_A\circ G_A|_{V_iM_kV_i}$ exists a unique normalized eigenvector $\gamma_i\in P_k$ such that $F_A\circ G_A(\gamma_i)=\gamma_i$ and $\Im(\gamma_i)=\Im(V_i)$, by lemma \ref{lemmairreducible}.

Notice that $\Im(\gamma_i)=\Im(V_i)\perp\Im(V_j)=\Im(\gamma_j)$, therefore $\gamma_1,\ldots,\gamma_n$ are orthonormal.
Complete this set to obtain an orthonormal basis $\{\gamma_1,\ldots,\gamma_n,\gamma_{n+1},\ldots,\gamma_{k^2}\}$ of $M_k$ formed by Hermitian matrices. Notice that  $F_A\circ G_A(\gamma_j)=0$, for $j>n$. Since $F_A$ and $G_A$ are adjoint maps, $G_A(\gamma_j)=0$ for $j>n$. 

Thus, $A=\gamma_1\otimes G_A(\gamma_1)+\ldots+\gamma_{k^2}\otimes G_A(\gamma_{k^2})= \gamma_1\otimes G_A(\gamma_1)+\ldots+\gamma_{n}\otimes G_A(\gamma_{n})$. Notice that $\langle G_A(\gamma_i), G_A(\gamma_j)\rangle=\langle \gamma_i, F_A\circ G_A(\gamma_j)\rangle=\langle \gamma_i, \gamma_j\rangle$, $1\leq i,j\leq n$, therefore $G_A(\gamma_1),\ldots, G_A(\gamma_n)$ are orthonormal too. Remind that $G_A$ is a positive map then $G_A(\gamma_i)\in P_m$. Define $\delta_i=G_A(\gamma_i)$.

Finally, if  $\sum_{i=1}^n\gamma_i'\otimes\delta_i'$ is a Hermitian Schmidt decomposition with $\gamma_i'\in P_{k},\delta_i'\in P_{m}$ then $F_A\circ G_A(\gamma_i')=\gamma_i'$. Thus, $F_A\circ G_A(V_i'M_kV_i')\subset V_i'M_kV_i'$, $1\leq i\leq n$, where $V_i'$ is the orthogonal projection onto $\Im(\gamma_i')$, by corollary \ref{corollarywellknown}. Notice that each $F_A\circ G_A|_{V_i'M_kV_i'}$ has one eigenvalue equal to 1 and the others equal to $0$. Thus, $F_A\circ G_A|_{V_i'M_kV_i'}$ is irreducible by lemma \ref{lemmairreducible}. Now, each $V_i'$ must be equal to some $V_j$, by proposition \ref{propositioncompletelyreducible}. 

Since each $F_A\circ G_A|_{V_jM_kV_j}=F_A\circ G_A|_{V_i'M_kV_i'}$ has only one eigenvalue equal to 1 then $\gamma_i'$ is a multiple of $\gamma_j$, but both matrices are positive semidefinite Hermitian matrices and normalized then $\gamma_i'=\gamma_j$. Thus, each $\gamma_i'$ is equal to some $\gamma_j$ and this Hermitian Schmdit decomposition is unique.
\end{proof}

\section{Examples and Counterexamples.}

In this section, we prove that if $A\in M_{k}\otimes M_m$ is positive under partial transposition (definition \ref{defPPT}) or symmetric with positive coefficients (definition \ref{defSPC}) or invariant under realignment (definition \ref{defmatricesinvariant}) then $F_A\circ G_A: M_k\rightarrow M_k$ is completely reducible. Thus, the theorems proved in the previous section hold for these three types of matrices. 
We shall search in section 6 for other types of matrices that could have the same property.

Here, we also present two examples of positive semidefinite Hermitian matrices in $M_{k}\otimes M_k\simeq M_{k^2}$ such that  $F_A\circ G_A: M_k\rightarrow M_k$ is not completely reducible.

\subsection{Examples}

Remind that $P_k$ stands for the set of positive semidefinite Hermitian matrices in $M_k$. Let $A\in M_k\otimes M_m\simeq M_{km}$, $A\in P_{km}$ and $A=\sum_{i=1}^n A_i\otimes B_i$. Let  $F_A: M_m\rightarrow M_k$ be $F_{A}(X)=\sum_{i=1}^ntr(B_iX)A_i$ and $G_A: M_k\rightarrow M_m$ be $G_{A}(X)=\sum_{i=1}^ntr(A_iX)B_i$. Recall that $tr(A(X\otimes Y))=tr(G_A(X)Y)=tr(XF_A(Y))$, for every $X\in M_k$ and $Y\in M_m$.  These are adjoint positive maps (see the introduction of section 3 for more details). Thus, $F_A\circ G_A:M_k\rightarrow M_k$ is a self-adjoint positive map.

\begin{theorem}\label{decomposablepropertyPPT}  Let $A\in M_k\otimes M_m\simeq M_{km}$, $A\in P_{km}$. If $A$ is PPT then $F_A\circ G_A:M_k\rightarrow M_k$ is completely reducible.
\end{theorem}
\begin{proof} Let $\gamma\in P_k\cap VM_kV$ be such that $F_A( G_A(\gamma))=\lambda\gamma$, $\lambda>0$. Let $V_1\in M_k$ be the orthogonal projection onto $\Im(\gamma)$.
Let $W_1\in M_m$ be the orthogonal projection onto $\Im(G_A(\gamma))$. 

By lemma \ref{lemmawellknown}, we have $G_A(V_1M_kV_1)\subset W_1M_mW_1$ and $F_A(W_1M_mW_1)\subset V_1M_kV_1$.

If $V_2=Id-V_1$ and $W_2=Id-W_1$ then $A=\sum_{i,j,k,s=1}^2 (V_i\otimes W_j)A(V_k\otimes W_s)$.

Notice that $tr(A (V_1\otimes W_2))=tr(G_A(V_1)W_2)=0$. Thus,  $A (V_1\otimes W_2)=(V_1\otimes W_2)A=0$, since $A\in P_{km}$ and $V_1\otimes W_2\in P_{km}$.
 Notice that $tr(A (V_2\otimes W_1))=tr(V_2F_A(W_1))=0$. Thus, $A (V_2\otimes W_1)=(V_2\otimes W_1)A=0$, since $A\in P_{km}$ and $V_2\otimes W_1\in P_{km}$.
 
Therefore,  $A=\sum_{i,j=1}^2 (V_i\otimes W_i)A(V_j\otimes W_j)$.

Next, $0=(A (V_1\otimes W_2))^{t_2}=(Id\otimes W_2^t)A^{t_2}(V_1\otimes Id)$ and $0=tr((Id\otimes W_2^t)A^{t_2}(V_1\otimes Id))=tr(A^{t_2}(V_1\otimes W_2^{t}))$. Since $A$ is PPT then $A^{t_2}$ is positive semidefinite and $A^{t_2}(V_1\otimes W_2^{t})=(V_1\otimes W_2^{t})A^{t_2}=0$.  Analogously, we obtain $A^{t_2}(V_2\otimes W_1^{t})=(V_2\otimes W_1^{t})A^{t_2}=0$. 

Thus, $A^{t_2}=\sum_{i,j=1}^2 (V_i\otimes W_j^t)A^{t_2}(V_j\otimes W_i^t)$ and $A^{t_2}=\sum_{i=1}^2 (V_i\otimes W_i^t)A^{t_2}(V_i\otimes W_i^t)$.

Finally, $A=\sum_{i=1}^2 (V_i\otimes W_i)A(V_i\otimes W_i)$ and by lemma \ref{lemmaequivalence}, $F_A\circ G_A:M_k\rightarrow M_k$ is completely reducible.
\end{proof}

\begin{theorem}\label{decomposablepropertySPC} Let $A\in M_k\otimes M_k\simeq M_{k^2}$, $A\in P_{k^2}$. If $A$ is SPC then $F_A\circ G_A:M_k\rightarrow M_k$ is completely reducible.
\end{theorem}
\begin{proof} If $A$ is SPC then $A=\sum_{i=1}^n\lambda_i\gamma_i\otimes\gamma_i$ is a Hermitian Schmidt decomposition of $A$, with $\lambda_i>0$, by definition \ref{defSPC}. Thus, $G_A: M_k\rightarrow M_k$, $G_{A}(X)=\sum_{i=1}^n\lambda_i\gamma_itr(\gamma_iX)$, is self-adjoint with non-negative eigenvalues  (the eigenvalues are $\lambda_i$ or $0$), which implies that the eigenvectors of $G_A$ and $G_A^2=F_A\circ G_A$ are the same.

Let $\gamma\in P_k\cap VM_kV$ be such that $F_A( G_A(\gamma))=\lambda^2\gamma$, $\lambda>0$. Thus, $G_A(\gamma)=\lambda\gamma$. Let $V_1\in M_k$ be the orthogonal projection onto $\Im(\gamma)$. 
By lemma \ref{lemmawellknown}, we have $G_A(V_1M_kV_1)\subset V_1M_kV_1$.

If $V_2=Id-V_1$ then $A=\sum_{i,j,k,s=1}^2 (V_i\otimes V_j)A(V_k\otimes V_s)$.

Notice that $tr(A (V_1\otimes V_2))=tr(G_A(V_1)V_2)=0$. Thus,  $A (V_1\otimes V_2)=(V_1\otimes V_2)A=0$, since $A\in P_{k^2}$ and $V_1\otimes V_2\in P_{k^2}$.
 Notice that $tr(A (V_2\otimes V_1))=tr(V_2F_A(V_1))=0$. Thus, $A (V_2\otimes V_1)=(V_2\otimes V_1)A=0$, since $A\in P_{k^2}$ and $V_2\otimes V_1\in P_{k^2}$.

Therefore,  $A=\sum_{i,j=1}^2 (V_i\otimes V_i)A(V_j\otimes V_j)$.

Next, $0=(A (V_1\otimes V_2))^{t_2}=(Id\otimes V_2^t)A^{t_2}(V_1\otimes Id)$ and $0=S((Id\otimes V_2^t)A^{t_2}(V_1\otimes Id))=(Id\otimes V_1^t)S(A^{t_2})(V_2\otimes Id)$, by property 3 of lemma \ref{propertiesofS}.

Now, $0=tr((Id\otimes V_1^t)S(A^{t_2})(V_2\otimes Id))=tr(S(A^{t_2})(V_2\otimes V_1^t))$. Since $S(A^{t_2})$ is positive semidefinite, by lemma \ref{SPCequiv}, $S(A^{t_2})(V_2\otimes V_1^t)=(V_2\otimes V_1^t)S(A^{t_2})=0$.  Analogously, we obtain $S(A^{t_2})(V_1\otimes V_2^t)=(V_1\otimes V_2^t)S(A^{t_2})=0$.

Thus, $A^{t_2}=\sum_{i,j=1}^2 (V_i\otimes V_j^t)A^{t_2}(V_j\otimes V_i^t)$ and $S(A^{t_2})=\sum_{i,j=1}^2 (V_i\otimes V_j^t)S(A^{t_2})(V_j\otimes V_i^t)=\sum_{i=1}^2 (V_i\otimes V_i^t)S(A^{t_2})(V_i\otimes V_i^t)$, by property 3 of lemma \ref{propertiesofS}.

Finally, $A^{t_2}=S^2(A^{t_2})=\sum_{i=1}^2 (V_i\otimes V_i^t)S^2(A^{t_2})(V_i\otimes V_i^t)=\sum_{i=1}^2 (V_i\otimes V_i^t)A^{t_2}(V_i\otimes V_i^t)$, by properties 2 and 3 of lemma \ref{propertiesofS}. Therefore,
$A= \sum_{i=1}^2 (V_i\otimes V_i)A(V_i\otimes V_i)$. Thus, by lemma \ref{lemmaequivalence}, $F_A\circ G_A:M_k\rightarrow M_k$ is completely reducible.
\end{proof}

\begin{theorem}\label{decomposablepropertyInvariant} Let $A\in M_k\otimes M_k\simeq M_{k^2}$, $A\in P_{k^2}$. If $A$ is invariant under realignment then $F_A\circ G_A:M_k\rightarrow M_k$ is completely reducible.
\end{theorem}
\begin{proof}
Let $A$ be invariant under realigment and 
 let $\sum_{i=1}^n\lambda_i\gamma_i\otimes\gamma_i^t$ be a Hermitian Schmidt decomposition of $A$, $\lambda_i>0$ for every $i$, by lemma \ref{Formatinvariant}. 
 
Notice that $A^{t_2}= \sum_{i=1}^n\lambda_i\gamma_i\otimes\gamma_i$ and $G_{A^{t_2}}: M_k\rightarrow M_k$, $G_{A^{t_2}}(X)=\sum_{i=1}^n\lambda_i\gamma_itr(\gamma_iX)$, is self-adjoint with non-negative eigenvalues (the eigenvalues are $\lambda_i$ or $0$), which implies that the eigenvectors of $G_{A^{t_2}}$ and $G_{A^{t_2}}^2=F_{A^{t_2}}\circ G_{A^{t_2}}$ are the same. Notice also that $F_{A^{t_2}}\circ G_{A^{t_2}}(X)=F_A\circ G_A(X)=\sum_{i=1}^n\lambda_i^2tr(\gamma_iX)\gamma_i$.

Let $\gamma\in P_k\cap VM_kV$ be such that $F_A\circ G_A(\gamma)=\lambda^2\gamma$, $\lambda>0$. Thus, $F_{A^{t_2}}\circ G_{A^{t_2}}(\gamma)=G_{A^{t_2}}^2(\gamma)=\lambda^2\gamma$ then $G_{A^{t_2}}(\gamma)=\lambda\gamma$. 

Next, $G_{A^{t_2}}(X)^t=G_A(X)$ then $G_A(\gamma)=\lambda\gamma^t$ and $F_A(\gamma^t)=\lambda\gamma$. Let $V_1\in M_k$ be the orthogonal projection onto $\Im(\gamma)$. By lemma \ref{lemmawellknown}, we have $G_A(V_1M_kV_1)\subset V_1^tM_kV_1^t$ and $F_A(V_1^tM_kV_1^t)\subset V_1M_kV_1$.
 Now, if $V_2=Id-V_1$ then $A=\sum_{i,j,k,s=1}^2 (V_i\otimes V_j^t)A(V_k\otimes V_s^t)$.

Notice that $tr(A (V_1\otimes V_2^t))=tr(G_A(V_1)V_2^t)=0$. Thus,  $A (V_1\otimes V_2^t)=(V_1\otimes V_2^t)A=0$, since $A\in P_{k^2}$ and $V_1\otimes V_2^t\in P_{k^2}$.
 Notice that $tr(A (V_2\otimes V_1^t))=tr(V_2F_A(V_1^t))=0$. Thus, $A (V_2\otimes V_1^t)=(V_2\otimes V_1^t)A=0$, since $A\in P_{k^2}$ and $V_2\otimes V_1^t\in P_{k^2}$.  
 
Therefore,  $A=\sum_{i,j=1}^2 (V_i\otimes V_i^t)A(V_j\otimes V_j^t)$.

Next, $S(A)=\sum_{i,j=1}^2 S((V_i\otimes V_i^t)A(V_j\otimes V_j^t))=\sum_{i,j=1}^2 (V_i\otimes V_j^t)S(A)(V_i\otimes V_j^t)$, by property 3 of lemma \ref{propertiesofS}.  

Since $A=S(A)$, we have $A=\sum_{i,j=1}^2 (V_i\otimes V_j^t)A(V_i\otimes V_j^t)=\sum_{i=1}^2 (V_i\otimes V_i^t)A(V_i\otimes V_i^t)$.

Therefore, by lemma \ref{lemmaequivalence}, $F_A\circ G_A:M_k\rightarrow M_k$ is completely reducible. 
\end{proof}

\subsection{Counterexamples}

\begin{lemma}\label{notdecomposable} Let $u\in\mathbb{C}^k\otimes\mathbb{C}^k$ be the vector defined in \ref{definition1} and $A=uu^t\in M_k\otimes M_k$. The linear transformation $F_{A}\circ G_{A}: M_k\rightarrow M_k$ is not completely reducible.
\end{lemma}
\begin{proof}

By definition \ref{definition1}, $u=\sum_{i=1}^ke_i\otimes e_i$, where $\{e_1,\ldots,e_k\}$ is the canonical basis of $\mathbb{C}^k$. Thus, $A=uu^t=\sum_{i,j=1}^ke_ie_j^t\otimes e_ie_j^t$ and $G_A(X)=F_A(X)=\sum_{i,j=1}^ke_ie_j^t tr(e_ie_j^tX)=X^t$.

Now,  the identity map $Id=F_A\circ G_A:M_k\rightarrow M_k$ has null kernel and every matrix is an eigenvector. Thus, $Id:M_k\rightarrow M_k$ does not have the decomposition property (definition \ref{definitiondecompositionproperty}) and $F_A\circ G_A$ is not completely reducible by proposition \ref{propositioncompletelyreducible}.
\end{proof}

\begin{lemma}\label{notdecomposable2} Let $k\geq 3$. Let  $v_1,e_3\in\mathbb{C}^k$ be such that $v_1^t=(\frac{1}{\sqrt{2}},\frac{i}{\sqrt{2}},0,\ldots,0)$ and $e_3^t=(0,0,1,0,\ldots,0)$. Consider $v=v_1\otimes \overline{v_1}+e_3\otimes e_3\in\mathbb{C}^k\otimes\mathbb{C}^k$. Let $A$ be the positive semidefinite Hermitian matrix $A=v\overline{v}^t+S(\overline{v}v^t)\in M_k\otimes M_k$ then $F_{A}\circ G_{A}: M_k\rightarrow M_k$ is not completely reducible.
\end{lemma}

\begin{proof}

Let $\gamma_1=v_1\overline{v_1}^t,\gamma_2=\frac{e_3\overline{v_1}^t+v_1e_3^t}{\sqrt{2}}, \gamma_3=\frac{i(e_3\overline{v_1}^t-v_1e_3^t)}{\sqrt{2}},\gamma_4=e_3e_3^t$. Notice that $v\overline{v}^t=\sum_{i=1}^4\gamma_i\otimes\gamma_i^t$. 

Now, $S(\overline{v}v^t)=\overline{V}\otimes V$, where $V=F^{-1}(v)=v_1\overline{v_1}^t+e_3e_3^t$, by item 1 of lemma \ref{propertiesofS}. Thus, $A=\sum_{i=1}^4\gamma_i\otimes\gamma_i^t+\overline{V}\otimes V$.

Since  $0=tr(\gamma_1\gamma_2)=tr(\gamma_1\gamma_3)=tr(\gamma_1\gamma_4)=tr(\gamma_1\overline{V})$
then $G_A(\gamma_1)=\gamma_1^t$ and $F_A(\gamma_1^t)=\gamma_1$. Therefore $F_A\circ G_A(\gamma_1)=\gamma_1$.
Next,  $0=tr(\gamma_2\gamma_1)=tr(\gamma_2\gamma_3)=tr(\gamma_2\gamma_4)=tr(\gamma_2\overline{V})$. Thus, $G_A(\gamma_2)=\gamma_2^t$ and $F_A(\gamma_2^t)=\gamma_2$, thus $F_A\circ G_A(\gamma_2)=\gamma_2$.

Finally, notice that $\gamma_2\in (Id-V_1)M_kV_1\oplus V_1M_k(Id-V_1)=R$, where $V_1$ is the orthogonal projection onto $\Im(\gamma_1)$. Therefore $F_A\circ G_A|_{R}\neq 0$. Thus, $F_A\circ G_A$ does not have the decomposition property(definition \ref{definitiondecompositionproperty}) and $F_A\circ G_A$ is not completely reducible by proposition \ref{propositioncompletelyreducible}.
\end{proof}

\section{Last Application: The Last Mutually Unbiased Basis}

In this section, we obtain a new proof of the following theorem proved in \cite{weiner}: If there is a set of $k$ mutually unbiased bases of $\mathbb{C}^k$ then  exists another orthonormal basis which is mutually unbiased with the first $k$.
  Our proof relies on proposition \ref{mainapplication}. We also proved that this last basis is unique up to multiplication by complex numbers of norm 1.

\begin{proposition}\label{mainapplication}  Let $A\in M_k\otimes M_k\simeq M_{k^2}$, $A\in P_{k^2}$. If $A$ is invariant under realignment  and $F_A\circ G_A:M_k\rightarrow M_k$ has $n$ eigenvalues equal to $1$ and the others $0$ then 
\begin{itemize}
\item[a)] exists an orthonormal set $\{v_1,\ldots, v_n\}\subset\mathbb{C}^k$ such that  $A=\sum_{i=1}^nv_i\overline{v_i}^t\otimes\overline{v_i}v_i^t$. 
\item[b)] The orthonormal set of item a$)$ is unique up to multiplication by complex numbers of norm one.
\end{itemize}
\end{proposition}
\begin{proof}
Let $A$ be invariant under realignment and let $\sum_{i=1}^n\lambda_i A_i\otimes A_i^t$ be a Hermitian Schmidt decomposition of $A$, $\lambda_i>0$ for every $i$, by lemma \ref{Formatinvariant}. Since $\lambda_i^2$ are the non null eigenvalues of $F_A\circ G_A:M_k\rightarrow M_k$ associated to the eigenvectors $A_i$ then $\lambda_i=1$.
Thus, for every eigenvector $\gamma$ of $F_A\circ G_A:M_k\rightarrow M_k$ associated to $1$, we have $G_A(\gamma)=\gamma^t$.

By lemma \ref{decomposablepropertyInvariant}, $F_A\circ G_A:M_k\rightarrow M_k$ is completely reducible. By proposition  \ref{propositioneigenvalues1}, exists a unique Hermitian Schmidt decomposition of $A$ , $\sum_{i=1}^n\gamma_i\otimes\delta_i$, such that $\gamma_i\in P_{k},\delta_i\in P_{k}$ for $1\leq i\leq n$. Notice that $\gamma_i^t=G_A(\gamma_i)=\delta_i$. Therefore, $A=\sum_{i=1}^n\gamma_i\otimes\gamma_i^t$ is the unique Hermitian Schmidt decomposition of $A$ such $\gamma_i\in P_{k}$ for $1\leq i\leq n$.

Let $V_i$ be the orthogonal projection onto $\Im(\gamma_i)$. Since $\{\gamma_1,\ldots,\gamma_n\}$ is an orthonormal set and each $\gamma_i\in P_k$ then $\Im(V_i)\perp\Im(V_j)$. Thus, $(V_i\otimes V_i^t)A(V_i\otimes V_i^t)=\gamma_i\otimes\gamma_i^t$. 

Now, let $F(\gamma_i)=r_i$(see definition \ref{definition1}). By definition \ref{defS}, $r_i\overline{r_i}^t=S(\gamma_i\otimes\overline{\gamma_i})$.
Since $\gamma_i$ is Hermitian then $\gamma_i\otimes\overline{\gamma_i}=\gamma_i\otimes\gamma_i^t$ and
$r_i\overline{r_i}^t=S(\gamma_i\otimes\gamma_i^t)=S((V_i\otimes V_i^t)A(V_i\otimes V_i^t))$.

Next,  by property 3 of lemma \ref{propertiesofS} $S((V_i\otimes V_i^t)A(V_i\otimes V_i^t))=(V_i\otimes V_i^t)S(A)(V_i\otimes V_i^t)$. 
Since $S(A)=A$ then $r_i\overline{r_i}^t=(V_i\otimes V_i^t)A(V_i\otimes V_i^t)=\gamma_i\otimes\gamma_i^t$.

Therefore, $\gamma_i\otimes\gamma_i^t$ has rank 1 and $\gamma_i$ has rank 1. Thus, $\gamma_i=v_i\overline{v_i}^t$ and $A=\sum_{i=1}^nv_i\overline{v_i}^t\otimes\overline{v_i}v_i^t$. 

 Since $tr(\gamma_i\gamma_j)=\delta_{ij}$ then $\{v_1,\ldots,v_n\}$ is an orthonormal set. 
 
 Finally,  suppose another orthonormal set $\{w_1,\ldots,w_n\}$ satisfy $A= \sum_{j=1}^n w_j\overline{w_j}^t\otimes \overline{w_j}w_j^t$. Since $\sum_{i=1}^n\gamma_i\otimes\gamma_i^t$ is unique then for each $p$ there is $q$ such that $w_p\overline{w_p}^t=v_q\overline{v_q}^t$. Therefore $w_p=cv_q$ with $|c|=1$.
\end{proof}

\begin{definition}\label{defMUBS}$($\textbf{Mutually Unbiased Bases}$)$ Let $\{v_1,\ldots,v_k\}$ and $\{w_1,\ldots,w_k\}$ be orthonormal bases of $\mathbb{C}^k$. We say that they are mutually unbiased if  $|\langle v_i,w_j\rangle |^2=\frac{1}{k}$ for every $i,j$. 
\end{definition}

\begin{definition} Let $\alpha=\{v_1,\ldots,v_k\}$ be an orthonormal basis of $\mathbb{C}^k$. Let us define $A_{\alpha}\in M_k\otimes M_k$ as $A_{\alpha}=\sum_{i=1}^kv_i\overline{v_i}^t\otimes\overline{v_i}v_i^t$. Notice that
$A_{\alpha}$ is invariant under realignment. 
\end{definition}

\begin{lemma} \label{mubscondition}If $\alpha,\beta$ are orthonormal basis of $\mathbb{C}^k$ then they are mutually unbiased if and only if $A_{\alpha}A_{\beta}=A_{\beta}A_{\alpha}=\frac{1}{k}uu^t$ $($Recall the definition of $u$ in \ref{definition1}$)$.
\end{lemma}
\begin{proof}
Let $\alpha=\{v_1,\ldots,v_k\}$, $\beta=\{w_1,\ldots,w_k\}$ and $A_{\alpha}=\sum_{i=1}^kv_i\overline{v_i}^t\otimes\overline{v_i}v_i^t$, $A_{\beta}=\sum_{j=1}^kw_j\overline{w_j}^t\otimes\overline{w_j}w_j^t$.

Now $A_{\alpha}A_{\beta}=\sum_{i,j=1}^kv_i\overline{w_j}^t\otimes\overline{v_i}w_j^t (\overline{v_i}^tw_j)(v_i^t\overline{w_j})$.

If $\alpha,\beta$ are mutually unbiased then for every $i,j$, we have $(\overline{v_i}^tw_j)(v_i^t\overline{w_j})=|\langle v_i,w_j\rangle |^2=\frac{1}{k}$. 

Therefore, $A_{\alpha}A_{\beta}=\sum_{i,j=1}^k\frac{1}{k}v_i\overline{w_j}^t\otimes\overline{v_i}w_j^t=\frac{1}{k}uu^t$, because $u=\sum_{i=1}^kv_i\otimes\overline{v_i}=\sum_{j=1}^k\overline{w_j}\otimes w_j$.

Now suppose that $A_{\alpha}A_{\beta}=\frac{1}{k}uu^t$. 

Therefore $\sum_{i,j=1}^k \lambda_{ij} v_i\overline{w_j}^t\otimes\overline{v_i}w_j^t =\frac{1}{k}uu^t$, where $\lambda_{ij}=|\langle v_i,w_j\rangle |^2$.

Next, $\frac{1}{k}Id\otimes Id=S(\frac{1}{k}uu^t)=\sum_{i,j=1}^k\lambda_{ij} v_i\overline{v_i}^t\otimes\overline{w_j}w_j^t$.

Notice that $\{v_i\otimes\overline{w_j},\ 1\leq i,j\leq k\}$ is an orthonormal basis of $\mathbb{C}^k\otimes\mathbb{C}^k$. Therefore $\lambda_{ij}$ are the eigenvalues of $\frac{1}{k}Id\otimes Id$, thus $\lambda_{ij}=\frac{1}{k}$.
\end{proof}

\begin{lemma} Let $\alpha_1,\ldots,\alpha_{k+1}$ be orthonormal bases of $\mathbb{C}^k$. If they are pairwise mutually unbiased then $\sum_{i=1}^{k+1}A_{\alpha_i}=Id\otimes Id+uu^t\in M_k\otimes M_k$.
\end{lemma}
\begin{proof}
Since $A_{\alpha_1},\ldots, A_{\alpha_{k+1}}$ commute, by lemma \ref{mubscondition}, there is a common basis of $\mathbb{C}^k\otimes\mathbb{C}^k$ formed by orthonormal eigenvectors. 
 Since $A_{\alpha_1},\ldots, A_{\alpha_{k+1}}$ are orthogonal projections and their pairwise multiplications are equal to $\frac{u}{\sqrt{k}}\frac{u^t}{\sqrt{k}}$,  by lemma \ref{mubscondition}, the intersection of their images is generated only by $u$. Notice that each $A_{\alpha_i}$ has rank $k$.
 
Thus, every $A_{\alpha_i}$ can be written as $\frac{u}{\sqrt{k}}\frac{u^t}{\sqrt{k}}+\sum_{l=(i-1)(k-1)+1}^{i(k-1)}r_l\overline{r_l}^t$, where $r_1,\ldots,r_{k^2-1},\frac{u}{\sqrt{k}}$ is a common orthonormal basis of  eigenvectors.

Finally, $\sum_{i=1}^{k+1}A_{\alpha_i}=(k+1)\frac{u}{\sqrt{k}}\frac{u^t}{\sqrt{k}}+\sum_{l=1}^{k^2-1}r_l\overline{r_l}^t=k\frac{u}{\sqrt{k}}\frac{u^t}{\sqrt{k}}+\frac{u}{\sqrt{k}}\frac{u^t}{\sqrt{k}}+\sum_{i=1}^{k^2-1}r_i\overline{r_i}^t=uu^t+Id\otimes Id$.
\end{proof}

\begin{theorem}\label{thelastbasis}
If $\mathbb{C}^k$ contains $k$ mutually unbiased bases then $\mathbb{C}^k$ contains $k+1$. Moreover, this last basis is unique up to multiplication by complex numbers of norm one.
\end{theorem}
\begin{proof} Let $\alpha_1,\ldots,\alpha_{k}$ be orthonormal bases of $\mathbb{C}^k$, which are pairwise mutually unbiased.
 Consider  $B=Id\otimes Id+uu^t-\sum_{i=1}^kA_{\alpha_i}$. 
 
 We saw in the proof of the previous lemma that we can write $Id\otimes Id+uu^t=(k+1)\frac{u}{\sqrt{k}}\frac{u^t}{\sqrt{k}}+\sum_{l=1}^{k^2-1}r_l\overline{r_l}^t$ and $\sum_{i=1}^kA_{\alpha_i}=k\frac{u}{\sqrt{k}}\frac{u^t}{\sqrt{k}}+\sum_{l=1}^{k(k-1)}r_l\overline{r_l}^t$, where $r_1,\ldots,r_{k^2-1},\frac{u}{\sqrt{k}}$ is an orthonormal basis of  $\mathbb{C}^k\otimes\mathbb{C}^k$. Thus, $B=\frac{u}{\sqrt{k}}\frac{u^t}{\sqrt{k}}+\sum_{l=k(k-1)+1}^{k^2-1}r_l\overline{r_l}^t$ is positive semidefinite with $k$ eigenvalues equal to 1 and the others zero. Notice that $BA_{\alpha_i}=A_{\alpha_i}B=\frac{u}{\sqrt{k}}\frac{u^t}{\sqrt{k}}$, $1\leq i\leq k$.
 
By lemma \ref{mubscondition}, in order to complete the proof, we need to show that $B=A_{\alpha_{k+1}}$ for some orthonormal basis $\alpha_{k+1}$ of $\mathbb{C}^k$.

Now, $S(B)=S(Id\otimes Id+uu^t-\sum_{i=1}^kA_{\alpha_i})=S(Id\otimes Id+uu^t)-\sum_{i=1}^kS(A_{\alpha_i})=Id\otimes Id+uu^t-\sum_{i=1}^kA_{\alpha_i}=B.$ Thus $B$ is invariant under realignment.

By lemma \ref{Formatinvariant}, we know that $B$ has a Hermitian Schmidt decomposition $\sum_{i=1}^n\lambda_i\gamma_i\otimes\gamma_i^t$ with $\lambda_i>0$. Therefore, $B=S(B)=\sum_{i=1}^n\lambda_iv_i\overline{v_i}^t$, where $v_i=F(\gamma_i)$. Since $F$ is an isometry, $\sum_{i=1}^n\lambda_iv_i\overline{v_i}^t$ is a spectral decomposition of $B$ and $\lambda_i$ are the eigenvalues. Then $n=k$ and $\lambda_i=1$, for $1\leq i\leq k$.

Thus, $\sum_{i=1}^k\gamma_i\otimes\gamma_i^t$ is a Hermitian Schmidit decomposition of $B$ and $F_B\circ G_B: M_k \rightarrow M_k$ is $F_B\circ G_B(X)=\sum_{i=1}^k tr(\gamma_i X)\gamma_i$, where $\gamma_1,\ldots,\gamma_k$ are orthonormal eigenvectors of $F_B\circ G_B$ associated to the eigenvalue $1$.
By proposition \ref{mainapplication}, exists an orthonormal basis $\alpha_{k+1}$ such that $B=A_{\alpha_{k+1}}$ and this basis is unique up to multiplication by complex numbers of norm one.
\end{proof}

\section{Chasing Completely Reducible Maps}

We saw in section 4 that if $A\in M_{k}\otimes M_m$ is positive under partial transposition  or symmetric with positive coefficients or invariant under realignment  then $F_A\circ G_A: M_k\rightarrow M_k$ is completely reducible.
In this section, we search  for other types of matrices that could have the same property.  Let $S_4$ be the group of permutations of $\{1,2,3,4\}$. 
We consider the  group $\{L_{\sigma}:M_k\otimes M_k\rightarrow M_k\otimes M_k, \sigma\in S_4\}$ acting on $M_k\otimes M_k$, we define two types of sets, $P_{\sigma}$, $I_{\sigma}$, and we consider the following two problems. 

Within these sets, we have the set of PPT matrices, the set of SPC matrices and the set of matrices invariant under realignment. We prove that these three types are the only types of matrices that appear in the solution of these two problems.

\begin{definition} Let $\sigma\in S_4$. Let $L_{\sigma}: M_k\otimes M_k\rightarrow M_k\otimes M_k$  be the linear transformation defined by $L_{\sigma}(a_1a_2^t\otimes a_3a_4^t)=a_{\sigma(1)}a_{\sigma(2)}^t\otimes a_{\sigma(3)}a_{\sigma(4)}^t,$ for every $a_1,a_2,a_3,a_4\in\mathbb{C}^k$.

Let $P_{\sigma}=\{A\in M_k\otimes M_k, A\ \text{and}\ L_{\sigma}(A)\ \text{are positive semidefinite}\}$ and 
$I_{\sigma}=\{A\in M_k\otimes M_k, A\ \text{is positive semidefinite and}\ A=L_{\sigma}(A)  \}$.
\end{definition}

\noindent
\textbf{Problem 1:} Which of the sets $P_{\sigma}$ contains only matrices $A$ such that $F_A\circ G_A:M_k\rightarrow M_k$ is completely reducible?\\\\
\textbf{Problem 2:} Which of the sets $I_{\sigma}$ contains only matrices $A$ such that $F_A\circ G_A:M_k\rightarrow M_k$ is completely reducible?

\subsection{Solution of Problem 1}

The solution of this problem is presented in the next theorem.
\begin{theorem}\label{solutionproblem1}
If for every $A\in P_{\sigma}$,  $F_A\circ G_A:M_k\rightarrow M_k$ is completely reducible then $P_{\sigma}=\{\text{PPT matrices}\}$ or $P_{\sigma}=\{\text{SPC matrices}\}$.
\end{theorem}
\begin{proof}
If $\mu=(12)(34)\in S_4$ and $\rho=(13)(24)\in S_4$ then $L_{\mu}(A)=A^t$ and $L_{\rho}(A)=TAT$, where $T$ is the flip operator $($definition \ref{definition1}$)$.

We can write 
$\displaystyle \{L_{\sigma},\ \sigma\in S_4\} =\bigcup_{i=1}^6 \{L_{\sigma_i}, L_{\mu}\circ L_{\sigma_i}, L_{\rho}\circ L_{\sigma_i},L_{\mu}\circ L_{\rho}\circ L_{\sigma_i} \}$, where $\sigma_1,\ldots,\sigma_6$ are the permutations in $S_4$ that fix 1. 
 These six permutations are associated to the following linear transformations $L_{\sigma_i}: M_k\otimes M_k\rightarrow M_k\otimes M_k$. 
 \begin{itemize}
 \item $L_{\sigma_1}(A)=A$
 \item  $L_{\sigma_2}(A)=A^{t_2}$, where $(a_1a_2^t\otimes a_3a_4^t)^{t_2}=a_1a_2^t\otimes a_4a_3^t$ $($the partial transposition$)$
 \item $L_{\sigma_3}(A)=S(A)$, where $S(a_1a_2^t\otimes a_3a_4^t)=a_1a_3^t\otimes a_2a_4^t$ $($the realignment map$)$
 \item $L_{\sigma_4}(A)=S(A^{t_2})$, where $S((a_1a_2^t\otimes a_3a_4^t)^{t_2})=a_1a_4^t\otimes a_2a_3^t$
 \item $L_{\sigma_5}(A)=S(A)^{t_2}$, where $S(a_1a_2^t\otimes a_3a_4^t)^{t_2}=a_1a_3^t\otimes a_4a_2^t$
 \item $L_{\sigma_6}(A)=AT$, where $(a_1a_2^t\otimes a_3a_4^t)T=a_1a_4^t\otimes a_3a_2^t$ ($T$ is the flip operator, definition  \ref{definition1})
 \end{itemize}  

Now, since each $L_{\mu}\circ L_{\sigma_i}(A)=L_{\sigma_i}(A)^t, L_{\rho}\circ L_{\sigma_i}(A)=TL_{\sigma_i}(A)T,L_{\mu}\circ L_{\rho}\circ L_{\sigma_i}(A)=(TL_{\sigma_i}(A)T)^t$ is positive semidefinite if and only if $L_{\sigma_i}(A)$ is positive semidefinite, then every  $P_{\sigma}$ is equal to some $P_{\sigma_i}$. 
So in order to solve problem 1, we must consider only the sets $P_{\sigma_i}$ such that $\sigma_i(1)=1$.

Notice that $L_{\sigma_1}(uu^t)=uu^t$, $L_{\sigma_3}(uu^t)=S(uu^t)=Id\otimes Id$, $L_{\sigma_5}(uu^t)=S(uu^t)^{t_2}=Id\otimes Id$ and $L_{\sigma_6}(uu^t)=uu^tT=uu^t$. Thus, $uu^t\in P_{\sigma_i},\ i=1,3,5,6$. By lemma \ref{notdecomposable},  $F_{uu^t}\circ G_{uu^t}:M_k\rightarrow M_k$ is not completely reducible. Thus, there are matrices $A$ within $P_{\sigma_1},P_{\sigma_3},P_{\sigma_5},P_{\sigma_6}$ such that $F_A\circ G_A:M_k\rightarrow M_k$ is not completely reducible.

Now, $P_{\sigma_2}=\{\text{PPT matrices}\}$ and $P_{\sigma_4}=\{\text{SPC matrices}\}$, by lemma \ref{SPCequiv}. By theorems \ref{decomposablepropertyPPT} and \ref{decomposablepropertySPC}, if $A$ is PPT or  SPC then $F_A\circ G_A:M_k\rightarrow M_k$ is completely reducible. Thus, if for every $A\in P_{\sigma}$ $F_A\circ G_A:M_k\rightarrow M_k$ is completely reducible then  $P_{\sigma}=\{\text{PPT matrices}\}$ or $\{\text{SPC matrices}\}$.
\end{proof}

\begin{remark} Remind that SPC matrices are PPT in $M_2\otimes M_2$ by theorem 4.3 in \cite{carielloSPC}. Thus, for $k=2$ the solution of problem $1$ is the set of PPT matrices.\\
\end{remark}

\subsection{Solution of Problem 2}

The solution of this problem is presented in the next two theorems.
\begin{theorem}\label{solutionproblem_2_k>2}
Let $k\geq 3$. If for every $A\in I_{\sigma}$,  $F_A\circ G_A:M_k\rightarrow M_k$ is completely reducible then $I_{\sigma}\subset \{\text{PPT matrices}\}$ or $I_{\sigma}\subset\{\text{SPC matrices}\}$or $I_{\sigma}\subset\{\text{Matrices Invariant under Realignment}\}$. 
\end{theorem}

The solution for $k=2$ is simpler.

\begin{theorem}\label{solutionproblem_2_k=2}
Let $k=2$. If for every $A\in I_{\sigma}$,  $F_A\circ G_A:M_k\rightarrow M_k$ is completely reducible then $I_{\sigma}\subset \{\text{PPT matrices}\}$.
\end{theorem}

\noindent\textit{Proof of theorem \ref{solutionproblem_2_k>2}}.

We saw in the solution of problem $1$ that
$\displaystyle \{L_{\sigma},\ \sigma\in S_4\} =\bigcup_{i=1}^6 \{L_{\sigma_i}, L_{\mu}\circ L_{\sigma_i}, L_{\rho}\circ L_{\sigma_i},L_{\mu}\circ L_{\rho}\circ L_{\sigma_i} \}$, where $\sigma_1,\ldots,\sigma_6$ are the permutations in $S_4$ that fix 1. Recall that $L_{\mu}(A)=A^t$, $L_{\rho}(A)=TAT$, where $T$ is the flip operator $($definition \ref{definition1}$)$ and $L_{\rho}(C\otimes D)=D\otimes C$, for every $C,D\in M_k$. 

 These permutations that fix $1$ are associated to the following linear transformations $L_{\sigma_i}: M_k\otimes M_k\rightarrow M_k\otimes M_k$. 
 \begin{itemize}
 \item $L_{\sigma_1}(A)=A$
 \item  $L_{\sigma_2}(A)=A^{t_2}$, where $(a_1a_2^t\otimes a_3a_4^t)^{t_2}=a_1a_2^t\otimes a_4a_3^t$ $($the partial transposition$)$
 \item $L_{\sigma_3}(A)=S(A)$, where $S(a_1a_2^t\otimes a_3a_4^t)=a_1a_3^t\otimes a_2a_4^t$ $($the realignment map$)$
 \item $L_{\sigma_4}(A)=S(A^{t_2})$, where $S((a_1a_2^t\otimes a_3a_4^t)^{t_2})=a_1a_4^t\otimes a_2a_3^t$
 \item $L_{\sigma_5}(A)=S(A)^{t_2}$, where $S(a_1a_2^t\otimes a_3a_4^t)^{t_2}=a_1a_3^t\otimes a_4a_2^t$
 \item $L_{\sigma_6}(A)=AT$, where $(a_1a_2^t\otimes a_3a_4^t)T=a_1a_4^t\otimes a_3a_2^t$ ($T$ is the flip operator, definition  \ref{definition1})
 \end{itemize}

Next, let $A\in M_k\otimes M_k$ be a positive semidefinite Hermitian matrix. Recall that if $A$ is PPT or SPC or invariant under realignment then $F_A\circ G_A:M_k\rightarrow M_k$ is completely reducible, by theorems \ref{decomposablepropertyPPT}, \ref{decomposablepropertySPC} and \ref{decomposablepropertyInvariant}. Recall that, by lemma \ref{SPCequiv}, $A$ is SPC if and only if $S(A^{t_2})$ is also a positive semidefinite Hermitian matrix and $S^2=Id$, by property 2 in  lemma \ref{propertiesofS}.

Let us consider $A=L_{\sigma}(A)$ for each $\sigma\in S_4$ separately. Thus, we have to consider the following 24 cases.

\begin{itemize}
\item[1)] \underline{$A=L_{\sigma_1}(A)$:} Notice that $uu^t=uu^t$ and  $F_{uu^t}\circ G_{uu^t}:M_k\rightarrow M_k$ is not completely reducible by lemma \ref{notdecomposable}.
\item[2)] \underline{$A=L_{\mu}\circ L_{\sigma_1}(A)$: } Notice that $uu^t=(uu^t)^t$ and $F_{uu^t}\circ G_{uu^t}:M_k\rightarrow M_k$ is not completely reducible by lemma \ref{notdecomposable}.
\item[3)] \underline{$A=L_{\rho}\circ L_{\sigma_1}(A)$:} Notice that $uu^t=Tuu^tT$ and $F_{uu^t}\circ G_{uu^t}:M_k\rightarrow M_k$ is not completely reducible by lemma \ref{notdecomposable}.
\item[4)]  \underline{$A=L_{\mu}\circ L_{\rho}\circ L_{\sigma_1}(A)$:} Notice that $uu^t=(Tuu^tT)^t$ and $F_{uu^t}\circ G_{uu^t}:M_k\rightarrow M_k$ is not completely reducible by lemma \ref{notdecomposable}.
\item[5)] \underline{$A=L_{\sigma_2}(A)$:} If $A=A^{t_2}$ then $A$ is PPT.
\item[6)] \underline{$A=L_{\mu}\circ L_{\sigma_2}(A)$: }If $A=(A^{t_2})^t$ then $A^t=A^{t_2}$ and $A$ is PPT.
\item[7)] \underline{$A=L_{\rho}\circ L_{\sigma_2}(A)$:} If $A=T(A^{t_2})T$ then $TAT=A^{t_2}$ and $A$ is PPT.
\item[8)] \underline{$A=L_{\mu}\circ L_{\rho}\circ L_{\sigma_2}(A)$:} If $A=(T(A^{t_2})T)^t$ then $TA^tT=A^{t_2}$ and $A$ is PPT.
\item[9)] \underline{$A=L_{\sigma_3}(A)$:} If $A=S(A)$ then $A$ invariant under realignment.
\item[10)] \underline{$A=L_{\mu}\circ L_{\sigma_3}(A)$:} Let $A=v\overline{v}^t+S(\overline{v}v^t)\in M_k\otimes M_k$, $k\geq 3$, as in lemma \ref{notdecomposable2}. 

Notice that by properties $(1)$ and $(2)$ in lemma \ref{propertiesofS} and since $V=F^{-1}(v)$ is Hermitian $($definition \ref{definition1}$)$, $L_{\mu}\circ L_{\sigma_3}(A)=S(A)^t=(V\otimes\overline{V}+\overline{v}v^t)^t=\overline{V}\otimes V+v\overline{v}^t=S(\overline{v}v^t)+v\overline{v}^t=A$.  By lemma \ref{notdecomposable2}, $F_{A}\circ G_{A}:M_k\rightarrow M_k$ is not completely reducible.
\item[11)]\underline{$A=L_{\rho}\circ L_{\sigma_3}(A)$:} Let $A=v\overline{v}^t+S(\overline{v}v^t)\in M_k\otimes M_k$, $k\geq 3$, as in lemma \ref{notdecomposable2}. 

Notice that by properties $(1)$ and $(2)$ in lemma \ref{propertiesofS} and since $L\overline{v}=v$ and $v^tT=\overline{v}^t$,  because $v\in\mathbb{C}^k\otimes\mathbb{C}^k$ is Hermitian $($definition \ref{definition1}$)$, $L_{\rho}\circ L_{\sigma_3}(A)=T(S(A))T=T(V\otimes\overline{V}+\overline{v}v^t)T=\overline{V}\otimes V+v\overline{v}^t=S(\overline{v}v^t)+v\overline{v}^t=A$.  By lemma \ref{notdecomposable2}, $F_{A}\circ G_{A}:M_k\rightarrow M_k$ is not completely reducible.
\item[12)] \underline{$A=L_{\mu}\circ L_{\rho}\circ L_{\sigma_3}(A)$:} If $A=(TS(A)T)^t$ then $TA^tT=S(A)$.

Thus, $S(A)$ is Hermitian and has a spectral decomposition $\sum_{i=1}^n\lambda_i v_i\overline{v_i}^t$ with $\lambda_i\in\mathbb{R}$. Thus, by properties $(1)$ and $(2)$ of lemma \ref{propertiesofS},  $A=\sum_{i=1}^n\lambda_i V_i\otimes \overline{V_i}$, where $F^{-1}(v_i)=V_i$. Therefore $Av$ is Hermitian for every Hermitian $v\in\mathbb{C}^k\otimes\mathbb{C}^k$. By lemma \ref{shape2},  $A$ has a Hermitian decomposition $\sum_{i=1}^n\alpha_i\gamma_i\otimes\gamma_i^t$. Thus, $T(A^t)T=T(\sum_{i=1}^n\alpha_i\gamma_i^t\otimes\gamma_i)T=\sum_{i=1}^n\alpha_i\gamma_i\otimes\gamma_i^t=A$.
Therefore, $A=S(A)$ and $A$ is invariant under realignment. 
\item[13)] \underline{$A=L_{\sigma_4}(A)$:} If $A=S(A^{t_2})$ then $A$ is SPC.
\item[14)] \underline{ $A=L_{\mu}\circ L_{\sigma_4}(A)$: }If $A=(S(A^{t_2}))^t$ then $A^t=S(A^{t_2})$ and $A$ is SPC.
\item[15)] \underline{ $A=L_{\rho}\circ L_{\sigma_4}(A)$:} If $A=T(S(A^{t_2}))T$ then $TAT=S(A^{t_2})$ and $A$ is SPC.
\item[16)] \underline{ $A=L_{\mu}\circ L_{\rho}\circ L_{\sigma_4}(A)$:} If $A=(TS(A^{t_2})T)^t$ then $TA^tT=S(A^{t_2})$ and $A$ is SPC.
\item[17)] \underline{ $A=L_{\sigma_5}(A)$:} If $A=S(A)^{t_2}$ then $S(A^{t_2})=A$ and $A$ is SPC. 
\item[18)] \underline{ $A=L_{\mu}\circ L_{\sigma_5}(A)$: }If $A=(S(A)^{t_2})^t$ then $A^{t_2}=S(A)^t=S(TAT)$, by property $(7)$ in lemma \ref{propertiesofS}. Thus, $S(A^{t_2})=TAT$  and $A$ is SPC.
\end{itemize}

In the next two cases, we shall need two simple properties of the partial transpositions, $(TBT)^{t_1}=TB^{t_2}T$ and $(B^{t_1})^{t_2}=B^t$, where $(\sum_{i}C_i\otimes D_i)^{t_1}=\sum_i C_i^t\otimes D_i$. 

\begin{itemize}
\item[19)] \underline{$A=L_{\rho}\circ L_{\sigma_5}(A)$:} If $A=TS(A)^{t_2}T$ then $A=(TS(A)T)^{t_1}=S(A^t)^{t_1}$, by property $(8)$ in lemma \ref{propertiesofS}. Therefore $A^{t_2}=S(A^t)^t=S(TA^tT)$, by property $(7)$ in lemma \ref{propertiesofS}, and $S(A^{t_2})=TA^tT$. Thus, $A$ is SPC.
\item[20)] \underline{$A=L_{\mu}\circ L_{\rho}\circ L_{\sigma_5}(A)$:} If $A=(TS(A)^{t_2}T)^t$ then $A=((TS(A)T)^{t_1})^t=(TS(A)T)^{t_2}=S(A^t)^{t_2}$. Thus, $S(A^{t_2})=A^t$ and $A$ is SPC.
\item[21)] \underline{$A=L_{\sigma_6}(A)$:} Notice that $uu^t=uu^tT$ and $F_{uu^t}\circ G_{uu^t}:M_k\rightarrow M_k$ is not completely reducible by lemma \ref{notdecomposable}.
\item[22)] \underline{$A=L_{\mu}\circ L_{\sigma_6}(A)$: } Notice that $uu^t=(uu^tT)^t$ and $F_{uu^t}\circ G_{uu^t}:M_k\rightarrow M_k$ is not completely reducible by lemma \ref{notdecomposable}.
\item[23)] \underline{$A=L_{\rho}\circ L_{\sigma_6}(A)$:} Notice that $uu^t=T(uu^tT)T$ and $F_{uu^t}\circ G_{uu^t}:M_k\rightarrow M_k$ is not completely reducible by lemma \ref{notdecomposable}.
\item[24)] \underline{$A=L_{\mu}\circ L_{\rho}\circ L_{\sigma_6}(A)$:} Notice that $uu^t=(Tuu^tTT)^t$ and $F_{uu^t}\circ G_{uu^t}:M_k\rightarrow M_k$ is not completely reducible by lemma \ref{notdecomposable}.
\end{itemize}

Finally, if for every $A\in I_{\sigma}$, $F_{A}\circ G_{A}:M_k\rightarrow M_k$ is completely reducible then
$I_{\sigma}\subset \{\text{PPT matrices}\}$ or $I_{\sigma}\subset\{\text{SPC matrices}\}$or $I_{\sigma}\subset\{\text{Matrices Invariant under Realignment}\}$. 
$\square$\\

In order to prove theorem \ref{solutionproblem_2_k=2}, we need the following lemma. In this lemma we show that  matrices invariant under realigment $($matrices of the case 9$)$, the matrices of the cases 10 and 11 of the proof of theorem \ref{solutionproblem_2_k>2} are PPT in $M_2\otimes M_2$. 
Notice that this lemma is not true in $M_k\otimes M_k$, for $k\geq 3$. In the remark \ref{counterexample}, we provide an example of a matrix invariant under realignment that is not PPT(or SPC) and in lemma \ref{notdecomposable2} we provide a counterexample for the other two types of matrices.

\begin{lemma}\label{invariantPPT} Let $A\in M_2\otimes M_2$ be a positive semidefinite Hermitian matrix. If $A=S(A)$ or $A^t=S(A)$ or $TAT=S(A)$ then A is PPT.
\end{lemma}
\begin{proof} First of all, remind that if $A\in M_2\otimes M_2$ is not PPT then $A^{t_2}$ has full rank and has only one negative eigenvalue. The reader can find this result in proposition 1 in \cite{augusiak}.

Now, if $S(A)=A$ or $A^t$ or $TAT$ then $S(A)$ is Hermitian and has a spectral decomposition $\sum_{i=1}^n\lambda_i v_i\overline{v_i}^t$ with $\lambda_i\in\mathbb{R}$. Thus, by properties $(1)$ and $(2)$ of lemma \ref{propertiesofS},  $A=\sum_{i=1}^n\lambda_i V_i\otimes \overline{V_i}$, where $F^{-1}(v_i)=V_i$. Therefore $Av$ is Hermitian for every Hermitian $v\in\mathbb{C}^k\otimes\mathbb{C}^k$. By lemma \ref{shape2},  $A$ has a Hermitian decomposition $\sum_{i=1}^n\alpha_i\gamma_i\otimes\gamma_i^t$.

Thus, $A^{t_2}=\sum_{i=1}^n\alpha_i\gamma_i\otimes\gamma_i$ and the subspaces of symmetric and anti-symmetric tensors in $\mathbb{C}^2\otimes\mathbb{C}^2$ are left invariant by $A^{t_2}$. Since $A^{t_2}$ is Hermitian and the subspace of anti-symmetric tensors in $\mathbb{C}^2\otimes\mathbb{C}^2$ is generated by $w=e_1\otimes e_2-e_2\otimes e_1$, where $\{e_1,e_2\}$ is the canonical basis of $\mathbb{C}^2$, then $A^{t_2}w=\lambda w$, where $\lambda\in\mathbb{R}$.

First, suppose $A=S(A)$. Thus, $A^{t_2}=S(AT)T=S(A)^{t_2}T=A^{t_2}T$, by properties $(4)$ and $(6)$ in lemma \ref{propertiesofS}. 
Since $T$ is the flip operator and $w$ is an anti-symmetric tensor then $Tw=-w$. Therefore $A^{t_2}w=A^{t_2}Tw=-A^{t_2}w$ and $A^{t_2}w=0$. Therefore $A^{t_2}$ does not have full rank and $A$ must be PPT. 

Second, suppose $A^t=S(A)$. Thus, $A^{t_2}=S(AT)T=S(A)^{t_2}T=(A^{t})^{t_2}T=(A^{t_2})^tT$, by properties $(4)$ and $(6)$ in lemma \ref{propertiesofS}. Since $A^{t_2}$ is hermitian, $(A^{t_2})^t=\overline{A^{t_2}}$ and $A^{t_2}=\overline{A^{t_2}}T$.
Since $\overline{w}=w$ and $\lambda\in\mathbb{R}$ then $\overline{A^{t_2}}w=\overline{A^{t_2}w}=\lambda w$. 
Thus, $\lambda w=A^{t_2}w=\overline{A^{t_2}}Tw=-\overline{A^{t_2}}w=-\lambda w$ and $\lambda=0$. Therefore $A^{t_2}$ does not have full rank and $A$ is PPT. 

Finally, if $TAT=S(A)$ then $A=TS(A)T=S(A^t)$, by property $(8)$ in lemma \ref{propertiesofS}. Thus, $S(A)=A^t$, by property $(2)$ in lemma \ref{propertiesofS} and by the last case $A$ is PPT.
\end{proof}

\begin{example}\label{counterexample} The matrix $Id\otimes Id+uu^t-T\in M_k\otimes M_k,\ k\geq 3,$ is invariant under realignment $($See \ref{example1}$)$ but it is not PPT, because its partial tranposition is $Id\otimes Id+T-uu^t$ and $(Id\otimes Id+T-uu^t)u=(2-k)u$. 
Notice also that $S((Id\otimes Id+uu^t-T)^{t_2})=S(Id\otimes Id+T-uu^t)=uu^t+T- Id\otimes Id$ and any anti-symmetric vector in 
$\mathbb{C}^k\otimes\mathbb{C}^k$ is an eigenvector of $uu^t+T- Id\otimes Id$ associated to $-2$. Thus, $S((Id\otimes Id+uu^t-T)^{t_2})$ is not positive semidefinite and $Id\otimes Id+uu^t-T$ is not SPC, by lemma \ref{SPCequiv}.
\end{example}

\noindent\textit{Proof of theorem \ref{solutionproblem_2_k=2}.}

First, by lemma \ref{invariantPPT}, the matrices $A$ considered in the cases 10 and 11 in the proof of theorem \ref{solutionproblem_2_k>2} are PPT for $k=2$, therefore $F_A\circ G_{A}:M_k\rightarrow M_k$ is completely reducible by lemma \ref{decomposablepropertyPPT}.

Now, the only cases in the proof of theorem \ref{solutionproblem_2_k>2} such that the hypothesis $k\geq 3$ was used, were the cases 10 and 11.
Thus, for the other cases we can use the same arguments used in the proof of theorem \ref{solutionproblem_2_k>2}. Just notice that SPC matrices and Matrices Invariant under Realigment are PPT for $k=2$, by theorem 4.3 in \cite{carielloSPC} and by lemma \ref{invariantPPT}, respectively. Therefore, if for every $A\in I_{\sigma}$, $F_A\circ G_{A}:M_k\rightarrow M_k$ is completely reducible then 
$I_{\sigma}\subset \{\text{PPT matrices}\}$.
$\square$

\section*{Summary}

In this paper we showed that completely reducible maps arise naturally in Quantum Information Theory.

In section 1, we described some properties of the realignment map, some preliminary results and definitions. In section 2, we defined the completely reducible property for a positive map and we showed that this property is equivalent to the decomposition property if the positive map is self-adjoint.

In section 3, we considered a positive semidefinite Hermitian matrix  $A=\sum_{i=1}^nA_i\otimes B_i\in M_k\otimes M_m\simeq M_{km}$  and the positive maps $F_A: M_m\rightarrow M_k$, $F_A(X)=\sum_{i=1}^n tr(B_iX)A_i$ and $G_A: M_k\rightarrow M_m$, $G_A(X)=\sum_{i=1}^n tr(A_iX)B_i$. These maps are adjoints with respect to the trace inner product. In this section, we proved that the self-adjoint map $F_A\circ G_A: M_k\rightarrow M_k$ is completely reducible if and only if $A$ has the property described in lemma \ref{lemmaequivalence}. This property is equivalent to the decomposition property discussed in  the previous section. 
In this section we assumed that $F_A\circ G_A: M_k\rightarrow M_k$ is completely reducible and we proved that $A$ is a sum of weakly irreducible matrices with support on orthogonal local Hilbert spaces. Thus, $A$ is separable if and only if each weakly irreducible summand is separable. We gave a completely description of weakly irreducible matrices in this case. We also showed that if the eigenvalues of $F_A\circ G_A: M_k\rightarrow M_k$ are $1$ or $0$ then $A$ is separable in a very strong sense.

In section 4, we showed that if $A$ positive under partial transposition or symmetric with positive coefficients or invariant under realignment then $A$ has the property described in lemma \ref{lemmaequivalence}, therefore $F_A\circ G_A: M_k\rightarrow M_k$ is completely reducible. Thus, all the theorems proved in section $3$ are valid for these three types of matrices. We also provide examples of positive semidefinite Hermitian matrices $A$ such that $F_A\circ G_A: M_k\rightarrow M_k$ is not completely reducible.

In section $5$, we obtained a new proof of the following fact: If $\mathbb{C}^k$ contains $k$ mutually unbiased basis then $\mathbb{C}^k$ contains $k+1$. Our proof relies on the fact that $A$ is separable, if the eigenvalues of $F_A\circ G_A: M_k\rightarrow M_k$ are $1$ or $0$ and  $A$ is invariant under realignment.

In section 6, we searched for other types of matrices that could have the same properties of PPT, SPC and matrices invariant under realignment without sucess. We considered a group of linear transformations acting on $M_k\otimes M_k$, which contains the partial transpositions and the realignment map. For each element of this group, we considered the set of matrices in $M_k\otimes M_k\simeq M_{k^2}$ that are positive and remain positive, or invariant, under the action of this element. Within this family of sets, we have the set of PPT matrices, the set of  SPC matrices and the set of  matrices invariant under realignment. We showed that these three sets are the only sets of this family such that $F_A\circ G_A: M_k\rightarrow M_k$  is completely reducible for every $A$ in the set. 

 It is interesting to notice that in $M_2\otimes M_2$, the PPT matrices is the only set of matrices, within this family of sets, such that $F_A\circ G_A: M_2\rightarrow M_2$  is completely reducible for every $A$ in the set. In order to obtain this theorem
we showed that every matrix invariant under realignment is PPT in $M_2\otimes M_2$ and we present a counterexample in $M_k\otimes M_k$, $k\geq 3$.

\vspace{12pt}
\textbf{Acknowledgement.}
D. Cariello was supported by CNPq-Brazil Grant 245277/2012-9.

\end{document}